\documentclass[journal,onecolumn,12pt]{IEEEtran} 
\usepackage{amsfonts}
\usepackage{amsmath}
\usepackage{amssymb}
\usepackage{graphicx}%
\usepackage{hyperref}
\providecommand{\U}[1]{\protect\rule{.1in}{.1in}}
\newtheorem{theorem}{Theorem}

\newtheorem{corollary}[theorem]{Corollary}

\let\originalleft\left
\let\originalright\right
\renewcommand{\left}{\mathopen{}\mathclose\bgroup\originalleft}
\renewcommand{\right}{\aftergroup\egroup\originalright}

\begin{document}

\title{Entanglement-assisted quantum turbo codes}
\author{Mark M. Wilde, Min-Hsiu Hsieh, and Zunaira Babar\thanks{Mark
M.~Wilde was 
with the School of Computer Science, McGill University, Montreal, Quebec,
Canada H3A 2A7 and is now with the Department of Physics and Astronomy and
the Center for Computation and Technology, Louisiana State University, Baton Rouge,
Louisiana, USA. Min-Hsiu Hsieh was with the ERATO-SORST Quantum Computation
and Information Project, Japan Science and Technology Agency 5-28-3, Hongo,
Bunkyo-ku, Tokyo, Japan and the Statistical
Laboratory, University of Cambridge, Wilberforce Road, Cambridge CB3 0WB, UK.
He is now with the Centre for Quantum Computation and Intelligent Systems,
Faculty of Engineering and Information Technology,
University of Technology, Sydney, P.O.~Box 123, Broadway NSW 2007, Australia.
Zunaira Babar is with the School of Electronics
and Computer Science, University of Southampton, SO17 1BJ, United Kingdom.
(E-mail: mwilde@gmail.com, minhsiuh@gmail.com, and zunaira.babar@gmail.com).

This paper was presented in part at the 2011 International Symposium on Information Theory in
Saint-Petersburg, Russia.}}
\date{\today}
\maketitle

\begin{abstract}
An unexpected breakdown in the existing theory of quantum serial
turbo coding is that a quantum convolutional encoder cannot simultaneously be
recursive and non-catastrophic. These properties are essential for quantum
turbo code families to have a minimum distance growing with blocklength and for their iterative
decoding algorithm to converge, respectively. Here, we show that the
entanglement-assisted paradigm simplifies the theory of quantum turbo codes, in the sense
that an entanglement-assisted quantum (EAQ)\ convolutional encoder can possess
both of the aforementioned desirable properties. We
give several examples of EAQ\ convolutional encoders that are both recursive
and non-catastrophic and detail their relevant parameters.
We then modify the quantum turbo decoding algorithm
of Poulin {\it et al.}, in order to have the constituent decoders pass along
only ``extrinsic information''
to each other rather than {\it a posteriori} probabilities
as in the decoder of
Poulin {\it et al.}, and this leads to a significant improvement in the
performance 
of unassisted quantum turbo codes.
Other simulation results
indicate that entanglement-assisted turbo codes can operate reliably in a
noise regime 4.73 dB beyond that of standard quantum turbo codes, when used on
a memoryless depolarizing channel. Furthermore,
several of our quantum turbo codes are within 1~dB or less of their hashing limits,
so that the performance of quantum turbo codes is now on par with
that of classical turbo codes.
Finally, we prove that entanglement
is \textit{the} resource that enables a convolutional encoder to be both
non-catastrophic and recursive because an encoder
acting on only information qubits, classical
bits, gauge qubits, and ancilla qubits cannot simultaneously satisfy them.
\end{abstract}


\section{Introduction}

Classical turbo codes represent one of the great successes of the modern
coding era~\cite{BGT93,BM96,BG96,BDMP98}. These near Shannon-limit codes have
efficient encodings, they offer astounding performance on memoryless channels,
and their iterative decoding algorithm quickly converges to an accurate error
estimate. They are ``probabilistic codes,'' meaning that they possess sufficient
structure to ensure efficient encoding and decoding, yet they have enough
randomness to allow for analysis of their performance with the probabilistic
method~\cite{BM96,BDMP98,KU98}.

The theory of quantum turbo codes is much younger than its classical
counterpart~\cite{PTO09}, and we still stand to learn more
regarding these codes' performance and structure. Poulin \textit{et al}.~set
this theory on a firm foundation \cite{PTO09} in an attempt to construct
explicit quantum codes that come close to achieving the quantum capacity of a
quantum
channel~\cite{PhysRevA.55.1613,capacity2002shor,ieee2005dev,qcap2008first}.
The structure of a quantum serial turbo code is similar to its classical
counterpart---one quantum convolutional encoder
\cite{PhysRevLett.91.177902,ieee2007forney}\ followed by a quantum interleaver
and another quantum convolutional encoder. The encoder \textquotedblleft
closer to the channel\textquotedblright\ is the inner encoder, and the one
\textquotedblleft farther from the channel\textquotedblright\ is the outer
encoder. One of the insights of Poulin \textit{et al}.~in Ref.~\cite{PTO09}
was to \textquotedblleft quantize\textquotedblright\ the classical notion of a
state diagram~\cite{V71,VVS97}---this diagram helps in analyzing important
properties of the constituent quantum convolutional encoders that directly
affect the performance of the resulting quantum turbo code.

Despite Poulin \textit{et al}.'s~success in providing a solid theoretical
construction, they discovered an unexpected
breakdown in the theory of quantum turbo codes. They found that
quantum convolutional encoders cannot be simultaneously non-catastrophic and
recursive, two desirable properties that can hold simultaneously for classical
convolutional encoders and are one reason underpinning the high performance of
classical turbo codes~\cite{BM96,BDMP98}. These two respective properties
ensure that an iterative decoder performs well in estimating errors and that
the turbo code family has a minimum distance growing almost linearly with
the length of the code~\cite{KU98,P09,OPT08}. Quantum convolutional encoders
cannot have these properties simultaneously, essentially because stabilizer
operators must satisfy stringent commutativity constraints in order to form a
valid quantum code (see Theorem~1 of Ref.~\cite{PTO09} or the simplified proof
in Ref.~\cite{HW12}). Thus, the existing
quantum turbo codes with non-catastrophic constituent quantum convolutional
encoders do not have a growing minimum distance, but Poulin \textit{et
al}.~conducted numerical simulations and showed that performance of their
quantum turbo codes appears to be good in practice.

The breakdown in the quantum turbo coding theory has led researchers to ponder
if some modification of the quantum turbo code construction could have
both a growing minimum distance and the iterative decoding algorithm
converging~\cite{T09}. One possibility is simply to change the paradigm for
quantum error correction, by allowing the sender and the receiver access to
shared entanglement before communication begins. This paradigm is known as the
\textquotedblleft entanglement-assisted\textquotedblright\ setting, and it
simplifies both the theory of quantum error correction \cite{BDH06,DBH09} and
the theory of quantum channels \cite{PhysRevLett.83.3081,ieee2002bennett}. In
entanglement-assisted quantum (EAQ) error correction, it is not necessary for
a set of stabilizer operators to satisfy the stringent commutativity
constraints that a standard quantum code should satisfy, allowing us to
produce EAQ codes from arbitrary classical codes.\footnote{This holds for EAQ
convolutional codes in addition to EAQ block codes~\cite{WB07,WB08}.} In
the theory of quantum channels, the classical and quantum capacities of a quantum
channel assisted by entanglement are the only known capacities for which we
can claim a complete understanding in the general case---both have formulas
involving an optimization of the quantum mutual information with respect to
all pure, bipartite entangled inputs to \textit{a single use} of the channel
(formally analogous to Shannon's formula for the classical capacity of a
classical channel~\cite{bell1948shannon}). The simplification that
entanglement offers to the theory of quantum information has led to the
following musing of Hayden \textit{et al}.~\cite{DHL10}:

\begin{quote}
\textquotedblleft To what extent does the addition of free entanglement make
quantum information theory similar to classical information
theory?\textquotedblright
\end{quote}

A naive attempt at constructing entanglement-assisted quantum turbo codes
(EAQTCs)
would be to produce them from classical turbo codes simply by following the recipe
given in Refs.~\cite{BDH06,DBH09}. That is, one could use the parity check matrix
of a classical turbo code to build an EAQTC according to the well-known CSS
construction \cite{PhysRevA.54.1098,PhysRevLett.77.793} and its
entanglement-assisted generalization \cite{BDH06,DBH09}. Though, this approach suffers from
several drawbacks, which are not present when one constructs EAQTCs from first principles:
\begin{itemize}

\item Following the recipe of Refs.~\cite{BDH06,DBH09} is
really just a ``blind import,'' and as such, it excludes 
us from understanding the theory of EAQTCs
at a deeper level. As discussed before, there are important theoretical issues
with the theory of quantum turbo coding \cite{PTO09}, and understanding the state diagram
of an EAQTC could in turn be helpful for understanding issues having to do with
recursiveness and non-catastrophicity.

\item The ``blind import'' approach does not provide any insight for achieving
an encoding efficiency beyond the $O(n^2)$ efficiency given by the encoding
algorithms from Refs.~\cite{thesis97gottesman,grassl2006itw}. In comparison,
the ``first principles'' approach given here leads to an encoder with a complexity linear in
the block length $n$. Also, a first principles approach gives control over the number of memory
qubits used by the constituent quantum convolutional encoders \cite{HHW12}, and this is
an important parameter contributing to the complexity of the encoder and the decoding algorithm.\footnote{In
our statement that a first principles approach leads to a linear encoding complexity, note that we are fixing
the number of the memory qubits to be constant with respect to the blocklength.
Of course, in any practical setting, minimizing the
number of memory qubits is essential because the complexity of the encoder and decoder
grows exponentially with the number of memory qubits.}

\item The ``blind import'' approach does not provide any insight for constructing
a decoding algorithm to take advantage of important effects such
as degeneracy \cite{SS96,PhysRevA.57.830}, nor is it clear that
the decoding algorithm will be as efficient as one could have from a first principles approach. 
Indeed, the first principles approach given here leads to a decoding algorithm (based on that from
Ref.~\cite{PTO09}) with a complexity linear in the block length.

\item The ``blind import'' approach does not give any clear control over the
entanglement consumption rate of the resulting EAQTC, other than that which is given by the formulas
in Refs.~\cite{arxiv2007brun,arx2008wildeOEA}. Given that shared entanglement is a precious resource,
it would be desirable to minimize the consumption of it. The first principles approach outlined here
gives the quantum code designer precise control over the entanglement consumption rate of the resulting
EAQTC.

\end{itemize}
Clearly, given all of the above, it is a worthwhile endeavor to construct a theory
of EAQTCs from first principles.

\section{Summary of Results}

In this paper, we show that entanglement assistance simplifies the theory
of quantum turbo codes in several important ways, we significantly enhance the performance
of the quantum turbo decoding algorithm from Ref.~\cite{PTO09},
and we also examine the effect on the performance
of quantum turbo codes by adding entanglement assistance. Specifically,
\begin{enumerate}

\item We develop a ``first principles'' approach to entanglement-assisted quantum turbo codes.
Although this theory
is admittedly a straightforward extension of the theory of entanglement-assisted codes
\cite{BDH06,DBH09} and quantum turbo codes \cite{PTO09}, it is necessary for us to develop it in order
to understand how notions such as the state diagram, non-catastrophicity, and recursiveness change
in the entanglement-assisted setting.

\item We show how to circumvent the ``no-go'' theorem of Ref.~\cite{PTO09}. In particular,
we find many
examples of EAQ convolutional encoders that can simultaneously be recursive
and non-catastrophic. 

\item We enhance the performance of the quantum turbo decoding algorithm
from Ref.~\cite{PTO09}, by having the constituent decoders pass along
``extrinsic information'' to each other rather than {\it a posteriori} probabilities
as in the decoding algorithm of Ref.~\cite{PTO09}. This modification is consistent
with how classical turbo decoding algorithms operate and is one of the reasons
why they perform near the Shannon limit. In particular,
several of our quantum turbo codes are within 1~dB or less of their hashing limits,
so that the performance of quantum turbo codes is now on par with
that of classical turbo codes.

\item Our simulations explore the effects of adding entanglement assistance
in various ways to unassisted quantum turbo codes. The results
of these simulations indicate that adding entanglement assistance
increases their performance on a memoryless
depolarizing channel (as one would expect), but they also suggest how to make judicious
use of entanglement consumption in a quantum turbo code.
We also consider the more practical
situation in which the entanglement is noisy and find that particular
entanglement-assisted quantum turbo codes have a certain amount of robustness
to this noise.

\item We broaden the scope of the ``no-go'' theorem of Ref.~\cite{PTO09}
to quantum convolutional encoders acting on logical qubits, classical bits, ancilla
qubits, and gauge (mixed-state) qubits (i.e., we prove that all such encoders cannot be
both recursive and non-catastrophic).
This result implies that entanglement is \textit{the} resource enabling a
quantum convolutional encoder to be both recursive and non-catastrophic.

\item We finally explore how recursiveness, non-recursiveness, catastrophicity,
or non-catastrophicity
are preserved under various resource substitutions of a quantum convolutional
encoder, such as converting ancilla qubits to classical bits,
converting ancilla qubits to ebits, etc. This exploration reveals the relationships
underpinning different kinds of quantum convolutional encoders.

\end{enumerate}

The ability of an entanglement-assisted quantum turbo code to be simultaneously recursive
and non-catastrophic has important implications. A \textquotedblleft quantized\textquotedblright\ version
of the result in Ref.~\cite{KU98} implies that the quantum serial turbo
code family formed by employing such an encoder along with another non-catastrophic encoder has
a minimum distance growing with the length of the code~\cite{OPT08}, and non-catastrophicity implies
that it has good iterative decoding performance. This result for EAQ
convolutional encoders holds partly because all four Pauli operators acting on
half of a Bell state are distinguishable when performing a measurement on both
qubits in the Bell state (much like the super-dense coding
effect~\cite{PhysRevLett.69.2881}). The ability of EAQ convolutional encoders
to be simultaneously non-catastrophic and recursive is another way in which
shared entanglement aids in a straightforward quantization of a classical
result---this assumption thus simplifies and enhances the existing theory of
quantum turbo coding.

Regarding our simulations, we found two quantum convolutional
encoders that are comparable to the first and third encoders of Poulin
\textit{et al}.~\cite{PTO09}, in the sense that they have the same number of
memory qubits, information qubits, ancilla qubits, and a comparable distance
spectrum. These encoders are non-catastrophic, and they become recursive after
replacing all of the ancilla qubits with ebits. Additionally, the encoders
with full entanglement assistance have a distance spectrum much improved over
the unassisted encoders, essentially because entanglement increases the
ability of a code to correct errors~\cite{LB10}.

We constructed a quantum
serial turbo code with these encoders and conducted four types of
simulations:\ the first with the unassisted encoders, a second with full
entanglement assistance, a third with the inner encoder assisted, and a fourth
with the outer encoder assisted. Due to our enhancement of the
quantum turbo decoding algorithm from Ref.~\cite{PTO09},
the unassisted quantum turbo codes perform
significantly better than those in Ref.~\cite{PTO09}.
The encoders with full entanglement assistance have an improvement in
performance over the unassisted ones, in the sense that they can operate
reliably in a noise regime several dB beyond the unassisted turbo codes.
This is due to the improvement in the distance spectrum and is also due to
the encoder becoming recursive. Also, these codes come close to achieving the
entanglement-assisted hashing
bound~\cite{PhysRevLett.83.3081,PhysRevA.66.052313}, which is the ultimate
limit on their performance.
The quantum turbo codes with inner encoder
entanglement assistance have performance a few dB below the fully-assisted
code, but one advantage of them is that other simulations indicate that they
are more tolerant to noise on the ebits.

We organize this paper as follows. The first section establishes notations and
definitions for EAQ\ codes similar to those in Ref.~\cite{PTO09}, and it also
shows a way in which an EAQ code with only ebits is remarkably similar to a
classical code. Section~\ref{sec:EAQ-conv-defs}\ defines the state diagram of
an EAQ convolutional encoder---it reveals if an encoder is non-catastrophic
and recursive, and we review how to check for these properties.
Section~\ref{sec:examples} gives several examples of non-catastrophic,
recursive EAQ convolutional encoders and details their distance spectra. We
discuss the construction of an EAQ serial turbo code in
Section~\ref{sec:EAQ-turbo}\ and give several combinations of serial
concatenations that have good minimum-distance scaling.
In Section~\ref{sec:turbo-decode-mod}, we detail how to modify the quantum turbo decoding
algorithm from Ref.~\cite{PTO09} such that
the constituent decoders pass along
only extrinsic information, and we discuss why this leads to an improvement
in performance.
Section~\ref{sec:sim-results}\ contains our simulation results with
accompanying interpretations of them. In Section~\ref{sec:corollary}, we show
that entanglement is in fact \textit{the} resource that enables a
convolutional encoder to be both recursive and non-catastrophic---a corollary
of Theorem~1 in Ref.~\cite{PTO09}\ states that other resources such as
classical bits, gauge qubits, and ancilla qubits do not help. Section~\ref{sec:CE-EAQ} then
discusses encoders that act on information qubits, ancilla qubits, ebits, and classical bits,
and it gives an example of an encoders that can be recursive and non-catastrophic.
Section~\ref{sec:preserve} states some general observations regarding recursiveness and
non-catastrophicity for different types of encoders.
The conclusion
summarizes our contribution and states many open questions.

\section{EAQ Codes}

We first review some important ideas from the theory of EAQ codes in order to
prepare us for defining convolutional versions of them along with their
corresponding state diagrams. The development is similar to that in
Section~III of Ref.~\cite{PTO09}.

The encoder $V$\ of an EAQ code produces an encoded state by acting on a set
of $k$ information qubits in a state $\left\vert \psi\right\rangle $, $a$
ancilla qubits, and $c$ ebits:%
\[
\left\vert \overline{\psi}\right\rangle ^{AB}\equiv V^{A}(\left\vert
\psi\right\rangle ^{A}\otimes\left\vert 0_{a}\right\rangle ^{A}\otimes
\left\vert \Phi_{c}^{+}\right\rangle ^{AB}),
\]
where%
\begin{align*}
\left\vert 0_{a}\right\rangle ^{A}  &  \equiv\left\vert 0\right\rangle
^{\otimes a},\\
\left\vert \Phi_{c}^{+}\right\rangle ^{AB}  &  \equiv\left(  \left\vert
\Phi^{+}\right\rangle ^{AB}\right)  ^{\otimes c},\\
\left\vert \Phi^{+}\right\rangle ^{AB}  &  \equiv\frac{1}{\sqrt{2}}\left(
\left\vert 00\right\rangle ^{AB}+\left\vert 11\right\rangle ^{AB}\right)  ,\\
n  &  =k+a+c,
\end{align*}
and the sender Alice possesses $c$ halves of the entangled pairs while the
receiver Bob possesses the other $c$ halves (see Figure~\ref{fig:encoding}).
In what follows, we abuse notation by having $V$ refer to the
\textquotedblleft Clifford group\textquotedblright\ unitary operator that acts
as above, but having it also refer to a binary matrix that acts on binary
vectors---these binary vectors represent different Pauli operators that are
part of the specification of an EAQ code (e.g., see Ref.~\cite{DBH09}%
).\footnote{The representation of the encoder as a binary matrix leads to a
loss of global phase information. Though, this global phase information is not
important because measurement of the syndrome destroys it, and it is not
necessary for faithful recovery of the encoded state.}%

\begin{figure}
[ptb]
\begin{center}
\includegraphics[
natheight=3.893400in,
natwidth=3.927100in,
height=2.0159in,
width=2.0358in
]%
{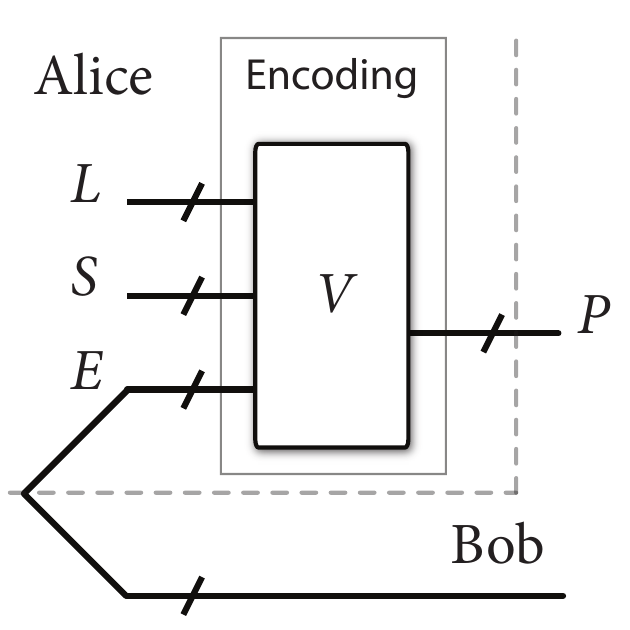}%
\caption{The encoder of an EAQ code. Alice acts on her logical qubits $L\,$,
local ancillas $S$, and her halves of the ebits $E$ with an encoding unitary
$V$. This encoding produces physical qubits $P$ that she then inputs to a
noisy quantum channel. The entanglement-assisted paradigm assumes that noise
does not occur on Bob's half of the ebits, but we later study this setting in
some simulations because it could occur in practice.}%
\label{fig:encoding}%
\end{center}
\end{figure}
Suppose now that Alice transmits her $n$ qubits of the encoded state
$\left\vert \overline{\psi}\right\rangle ^{AB}$ over a noisy Pauli channel.
Then the resulting state is $P^{A}\left\vert \overline{\psi}\right\rangle
^{AB}$ where $P^{A}$ is some $n$-fold tensor product of Pauli operators (in
what follows, we simply say an \textquotedblleft$n$-qubit Pauli
operator\textquotedblright). Suppose that Bob applies the inverse $\left(
V^{A}\right)  ^{\dag}$\ of the encoding to the state $P^{A}\left\vert
\overline{\psi}\right\rangle ^{AB}$. The resulting state has the following
form:%
\begin{align}
(V^{A})^{\dag}P^{A}\left\vert \overline{\psi}\right\rangle ^{AB}  &
=(V^{A})^{\dag}P^{A}V^{A}(\left\vert \psi\right\rangle ^{A}\otimes\left\vert
0_{a}\right\rangle ^{A}\otimes\left\vert \Phi_{c}^{+}\right\rangle
^{AB})\nonumber\\
&  =L^{A}\left\vert \psi\right\rangle ^{A}\otimes S^{A}\left\vert
0_{a}\right\rangle ^{A}\otimes E^{A}\left\vert \Phi_{c}^{+}\right\rangle
^{AB}), \label{eq:decoded-Paulis}%
\end{align}
where $L^{A}$ is some $k$-qubit Pauli operator, $S^{A}=S_{1}^{A}\otimes
\cdots\otimes S_{a}^{A}$ is an $a$-qubit Pauli operator, and $E^{A}=E_{1}%
^{A}\otimes\cdots\otimes E_{c}^{A}$ is some $c$-qubit Pauli operator. Observe
that%
\[
S^{A}\left\vert 0_{a}\right\rangle ^{A}=\left\vert s_{1}\right\rangle
^{A}\otimes\cdots\otimes\left\vert s_{a}\right\rangle ^{A},
\]
where $s_{i}=0$ if $S_{i}\in\left\{  I,Z\right\}  $ and $s_{i}=1$ otherwise.
The fact that $s_{i}$ is invariant under the application of a Pauli $Z$
operator implies that a quantum code can be degenerate (where two different
errors mapping to the same syndrome have the same effect on the encoded
state). Also, observe that%
\[
E^{A}\left\vert \Phi_{c}^{+}\right\rangle ^{AB}=\left\vert \Phi\left(
e_{1,x},e_{1,z}\right)  \right\rangle ^{AB}\otimes\cdots\otimes\left\vert
\Phi\left(  e_{c,x},e_{c,z}\right)  \right\rangle ^{AB},
\]
where $\left\vert \Phi\left(  e_{i,x},e_{i,z}\right)  \right\rangle ^{AB}$
denotes the four distinguishable Bell states and
\begin{align*}
\left(  e_{i,x},e_{i,z}\right)   &  =\left(  0,0\right)  \text{ if }E_{i}%
^{A}=I,\\
\left(  e_{i,x},e_{i,z}\right)   &  =\left(  0,1\right)  \text{ if }E_{i}%
^{A}=Z,\\
\left(  e_{i,x},e_{i,z}\right)   &  =\left(  1,0\right)  \text{ if }E_{i}%
^{A}=X,\\
\left(  e_{i,x},e_{i,z}\right)   &  =\left(  1,1\right)  \text{ if }E_{i}%
^{A}=Y.
\end{align*}
Thus all four Pauli operators are distinguishable in this case by performing a
Bell measurement (similar to the super-dense coding effect
\cite{PhysRevLett.69.2881}). Ebits do not contribute to the degeneracy of a
quantum code because different errors $E^{A}$ lead to distinct measurement results.

Bob can perform $Z$ basis measurements on the ancillas and Bell measurements
on the ebits to determine the syndrome $r\left(  P\right)  $ of the Pauli
error $P$:%
\[
r\left(  P\right)  \equiv\left(  s_{1},\ldots,s_{a},e_{1,x},e_{1,z}%
,\ldots,e_{c,x},e_{c,z}\right)  .
\]
Consider the following relation between the binary representations of $P$,
$V$, $L$, $S$, and $E$ in (\ref{eq:decoded-Paulis}):%
\[
PV^{-1}=\left(  L:S:E\right)  .
\]
The syndrome $r\left(  P\right)  $ only partially determines $S$, but it fully
determines $E$. Let us decompose the binary representation of $S$ as
$S=S^{z}+S^{x}$ and that of $E$ as $E=E^{z}+E^{x}$. When Bob performs his
measurements, he determines $S^{x}$, $E^{z}$, and $E^{x}$. That is, he
determines the following relations between the components $S_{i}^{x}$ of
$S^{x}$ and the components $s_{i}$ of the syndrome $r\left(  P\right)  $:%
\[
S_{i}^{x}=X\text{ if }s_{i}=1,\ \ \ \ S_{i}^{x}=I\text{ otherwise.}%
\]
The syndrome also determines the $i^{\text{th}}$ components $E_{i}^{x}$\ and
$E_{i}^{z}$ of $E$:%
\begin{align*}
E_{i}^{x}  &  =X\text{ if }e_{i,x}=1,\ \ \ \ E_{i}^{x}=I\text{ otherwise,}\\
E_{i}^{z}  &  =Z\text{ if }e_{i,z}=1,\ \ \ \ E_{i}^{z}=I\text{ otherwise.}%
\end{align*}

The phenomenon of degeneracy represents the most radical departure of quantum
coding from classical coding~\cite{SS96,PhysRevA.57.830}. Consider two
different physical errors $P$ and $P^{\prime}$ that differ only by $Z$
operators acting on the ancillas:%
\begin{align*}
P  &  =(L:S^{x}+S^{z}:E^{x}+E^{z})V\\
P^{\prime}  &  =(L:S^{x}+S^{\prime z}:E^{x}+E^{z})V\\
&  =P+(I_{k}:S^{z}+S^{\prime z}:I_{c})V,
\end{align*}
where $I_{k}$ is a length $2k$ zero vector (the binary representation of a
$k$-fold tensor product of identity operators). These different errors lead to
the same error syndrome. In the classical world, this would present a problem
for error correction. But this situation does not cause a problem in the
quantum world for the errors $P$ and $P^{\prime}$---the logical error
affecting the encoded quantum information is the same for both $P$ and
$P^{\prime}$, and Bob can correct either of these errors simply by applying
$L^{-1}$ after decoding.

We now define several sets of operators that are important for determining the
properties and performance of EAQ codes. The set $C\left(  I\right)  $\ of
harmless, undetected errors consists of all operators $P$\ that have a zero
syndrome, yet have no effect on the encoded state:%
\[
C\left(  I\right)  \equiv\left\{  P:P=(I_{k}:S^{z}:I_{c})V,\ \ S^{z}%
\in\left\{  I,Z\right\}  ^{a}\right\}  .
\]
This set of operators is equivalent to the isotropic subgroup, in the language
of Refs.~\cite{BDH06,DBH09}. It is also analogous to the all-zero codeword of
a classical code. The set $C\left(  L\right)  $ of harmful, undetected errors
consists of all operators $P$\ that have a zero syndrome, yet change the
quantum information in the encoded state:%
\[
C\left(  L\right)  \equiv\left\{  P:P=(L:S^{z}:I_{c})V,\ \ S^{z}\in\left\{
I,Z\right\}  ^{a}\right\}  ,
\]
where$\ L\neq I_{k}$. This set of operators corresponds to a single logical
transformation on the encoded state, depending on the choice of $L$, and it is
analogous to a single codeword of a classical code. The set $C\left(
L,S^{x},E^{x},E^{z}\right)  $ corresponds to a particular logical
transformation and syndrome, and is thus a logical coset:%
\begin{multline*}
C\left(  L,S^{x},E^{x},E^{z}\right) \\
\equiv\left\{  P:P=(L:S^{z}+S^{x}:E^{z}+E^{x})V,\ \ S^{z}\in\left\{
I,Z\right\}  ^{a}\right\}  .
\end{multline*}
It is analogous to a single erred codeword of a classical code if $S^{x}$,
$E^{x}$, or $E^{z}$ is non-zero. The operator codewords of an EAQ code belong
to the following set $C$:%
\begin{equation}
C\equiv\bigcup\limits_{L}C\left(  L\right)  , \label{eq:log-ops}%
\end{equation}
where $L$ is an arbitrary $k$-qubit Pauli operator. The set $C$ is equivalent
to the full set of logical operators for the code, and it is analogous to the
set of all codewords of a classical code. These definitions lead to the
definition of the minimum distance of an EAQ code as the minimum weight of an
operator $P$ in $C-C\left(  I\right)  $:%
\[
d\left(  C\right)  \equiv\min\left\{  w\left(  P\right)  :P\in C-C\left(
I\right)  \right\}  .
\]
This definition is similar to the definition of the distance of a classical
code, but it incorporates the coset structure of a quantum code. Thus, we can
determine the performance of the code in terms of distance by tracking its
logical operators, and this intuition is important when we move on to EAQ
convolutional codes. Also, the logical operators play an important role in
decoding because a maximum likelihood decoding algorithm for a quantum code
estimates the most likely logical error given the syndrome and a particular
physical noise model.\footnote{The definitions for the maximum likelihood
decoder of an EAQ code are nearly identical to those for stabilizer codes in
Section~IIIC of Ref.~\cite{PTO09}. Thus, we do not give them here.}

We end this section with a final remark concerning EAQ codes and the musing of
Hayden \textit{et al}.~in Ref.~\cite{DHL10}. Suppose that an EAQ code does not
exploit any ancilla qubits and uses only ebits. Then the features of the code
become remarkably similar to that of a classical code. Degeneracy, a uniquely
quantum feature of a code, does not occur in this case because the syndrome
completely determines the error, in analogy with error correction in the
classical world. Also, the code loses its coset structure, so that $C\left(
I\right)  $ is equal to the identity operator, $C\left(  L\right)  $ is
equivalent to just one logical operator, and the definition of the code's
minimum distance is the same as the classical definition.

\section{EAQ convolutional codes}

\label{sec:EAQ-conv-defs}An EAQ convolutional code is a particular type of EAQ
code that has a convolutional structure. In past work on this
topic~\cite{WB07,WB08,WB09}, we adopted the \textquotedblleft
Grassl-R\"{o}tteler\textquotedblright\ approach to this theory~\cite{GR06b},
by beginning with a mathematical description of the code and determining a
Grassl-R\"{o}tteler pearl-necklace encoder that can encode it. Here, we adopt
the approach of Poulin \textit{et al}.~\cite{PTO09} which in turn heavily
borrows from ideas in classical convolutional coding~\cite{V71}. We begin with
a seed transformation (a \textquotedblleft convolutional
encoder\textquotedblright\ or a \textquotedblleft quantum shift-register
circuit\textquotedblright) and determine its state diagram, which yields
important properties of the encoder. We can always rearrange a
Grassl-R\"{o}tteler pearl-necklace encoder as a convolutional encoder
\cite{W09,HHW10,HH10}, but it is not clear that every convolutional encoder
admits a form as a Grassl-R\"{o}tteler pearl necklace encoder. For this reason
and others, we adopt the Poulin \textit{et al}.~approach in what follows.

An EAQ convolutional encoder is a \textquotedblleft Clifford
group\textquotedblright\ unitary $U$ that acts on $m$\ memory qubits, $k$
information qubits, $a$ ancilla qubits, and $c$ halves of ebits to produce a
set of $m$ memory qubits and$\ n$ physical or channel
qubits,\footnote{Physical or channel qubits in the entanglement-assisted
paradigm are the ones that Alice transmits over the channel.} where $n=k+a+c$.
The transformation that it induces on binary representations of Pauli
operators acting on these registers is as follows:%
\[
\left(  M^{\prime}:P\right)  =(M:L:S:E)U,
\]
where $M^{\prime}$ acts on the $m$\ output memory qubits, $P$ acts on the $n$
output physical qubits, $M$ acts on the $m$ input memory qubits, $L$ acts on
the $k$ information qubits, $S$ acts on the $a$ ancilla qubits, and $E$ acts
on the $c$ halves of ebits.\ Although the quantum states in these registers
can be continuous in nature, the act of syndrome measurement discretizes the
errors acting on them, and the above classical representation is useful for
analysis of the code's properties and the flow of the logical operators
through the encoding circuit (recall that the goal of a decoding algorithm is
to produce good estimates of logical errors). This representation is similar
to the shift-register representation of classical convolutional codes, with
the difference that the representation there corresponds to the actual flow of
bit information through a convolutional circuit, while here it is merely a
useful tool for analyzing the properties of the encoder.

The overall encoding operation for the code is the transformation induced by
repeated application of the above seed transformation to a quantum data stream
broken up into periodic blocks of information qubits, ancilla qubits, and
halves of ebits while feeding the output memory qubits of one transformation
as the input memory qubits of the next (see Figure~6\ of Ref.~\cite{PTO09}%
\ for a visual aid). The advantage of a quantum convolutional encoder is that
the complexity of the overall encoding scales only linearly with the length of
the code for a fixed memory size, while the decoding complexity scales
linearly with the length of the code by employing a local maximum likelihood
decoder combined with a belief propagation algorithm~\cite{PTO09}. The quantum
communication rate of the code is essentially $k/n$ while the entanglement
consumption rate is $c/n$, if the length of the code becomes large compared to
$n$.

\subsection{State Diagram}

The state diagram of an EAQ\ convolutional encoder is the most important tool
for analyzing its properties, and it is the formal quantization of a classical
convolutional encoder's state diagram~\cite{V71,VVS97}. It examines the flow
of the logical operators through the encoder with a finite-state machine
approach, and this representation is important for analyzing both its distance
and its performance under the iterative decoding algorithm of
Ref.~\cite{PTO09}. The \textit{state diagram} is a directed multigraph with
$4^{m}$ vertices that we think of as memory states. We label each memory state
with an $m$-qubit Pauli operator $M$. We connect two vertices with a directed
edge from $M\rightarrow M^{\prime}$, labeled as $\left(  L,P\right)  $, if
there exists a $k$-qubit Pauli operator $L$, an $n$-qubit Pauli operator $P$,
and an $a$-qubit Pauli operator $S^{z}\in\left\{  I,Z\right\}  ^{a}$ such that%
\begin{equation}
\left(  M^{\prime}:P\right)  =(M:L:S^{z}:I_{c})U. \label{eq:state-diagram}%
\end{equation}
We refer to the labels $L$ and $P$ of an edge as the respective logical and
physical label.%
\begin{figure}
[ptb]
\begin{center}
\includegraphics[
natheight=1.826500in,
natwidth=2.806300in,
height=1.0075in,
width=1.6336in
]%
{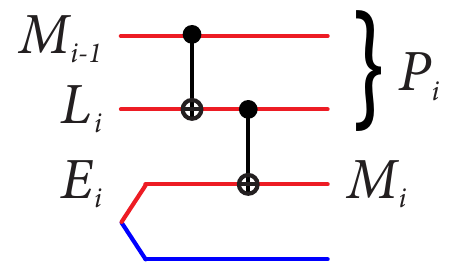}%
\caption{(Color online) Seed transformation for an $n=2$, $k=1$, $c=1$, $m=1$
EAQ convolutional code.}%
\label{fig:simple-encoder}%
\end{center}
\end{figure}

As an example, consider the transformation depicted in
Figure~\ref{fig:simple-encoder}. It acts on one memory qubit, one information
qubit, and one half of an ebit to produce two channel qubits and one memory
qubit. Figure~\ref{fig:state-diagram}\ illustrates the state diagram
corresponding to this transformation. There are four memory states because
there is only one memory qubit, and there are 16 edges because there are four
memory states and four logical operators for one information qubit and one
ebit.%
\begin{figure}
[ptb]
\begin{center}
\includegraphics[
natheight=1.753800in,
natwidth=1.483200in,
height=2.9836in,
width=2.5278in
]%
{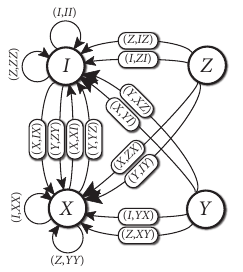}%
\caption{The state diagram corresponding to the seed transformation in
Figure~\ref{fig:simple-encoder}. The state diagram allows us to check that the
encoder is non-catastrophic and non-recursive.}%
\label{fig:state-diagram}%
\end{center}
\end{figure}

\subsection{Non-catastrophicity}

We now recall the definition of non-catastrophicity in Ref.~\cite{PTO09},
which is the formal quantization of the definition in Section~IX of
Ref.~\cite{V71}.\ The definition from Ref.~\cite{PTO09} is the same for an
EAQ\ convolutional encoder because it depends on the iterative decoding
algorithm used to decode the code, and we can exploit a slight variation of the iterative
decoding algorithm in Ref.~\cite{PTO09} to decode EAQ\ convolutional codes.
A \textit{path} through the state diagram is a sequence $M_{1}$, \ldots,
$M_{t}$ of vertices such that $M_{i}\rightarrow M_{i+1}$ is an edge belonging
to it. Each logical operator of $C$ in (\ref{eq:log-ops}) corresponds to a
path in the state diagram, with the sequence of vertices in the path being the
states of memory traversed while encoding the logical operator. The physical
and logical weights of a logical operator are equal to the sums of the
corresponding weights of the edges traversed in a path that encodes the
logical operator. A \textit{zero physical-weight cycle} is a cycle in the
state diagram such that all edges in the cycle have zero physical weight.
Finally, an EAQ\ encoder acting on $m$ memory qubits, $k$ information qubits,
$a$ ancilla qubits, and $c$ ebits is \textit{non-catastrophic} if every zero
physical-weight cycle in its state diagram has zero logical weight.

Why is this definition an appropriate definition of non-catastrophicity?
First, suppose that we modify the circuit
in Figure~\ref{fig:simple-encoder} so that the ebit is replaced
by an ancilla qubit (it thus becomes the same as Figure~8 of
Ref.~\cite{PTO09}). Such a replacement leads to a doubling of the number of
edges in the state diagram in our Figure~\ref{fig:state-diagram}\ because the
state diagram should include all of the transitions where a $Z$ operator acts
on the ancilla qubits (these all lead to other logical operators in their
case). The new state diagram (see Figure~9 of Ref.~\cite{PTO09}) then includes a
self-loop at memory state $Z$ with zero physical weight and non-zero logical
weight, and it is thus catastrophic according to the definition. The problem
with such a loop is that it can \textquotedblleft throw off\textquotedblright%
\ the iterative decoding algorithm. If a $Y$ error were to act on the second physical qubit in
one frame of the code (while the identity acts on the rest), then
\textquotedblleft pushing\textquotedblright\ this error through the inverse of
the encoder applies an $X$ operator to the ancilla qubit and produces one
syndrome bit for that frame of the code, but it propagates $Z$ errors onto
every logical qubit in the stream while applying $Z$ errors to every memory
qubit. All of these other errors go undetected by an iterative decoder because
the error propagation does not trigger any additional syndrome bits besides
the initial one that the $Y$ error triggered.

Observe that the state diagram in Figure~\ref{fig:state-diagram}\ for the
EAQ\ convolutional encoder does not feature a zero physical-weight cycle with
non-zero logical weight. Thus, the encoder is non-catastrophic, illustrating
another departure from the classical theory of turbo codes. \textit{Non-catastrophicity in the
quantum world is not only a property of the encoder, but it also depends on
the resources available for encoding.} If we analyze the above scenario that
leads to catastrophic error propagation in the unassisted encoder,
we see that it does not lead to such propagation for the entanglement-assisted
encoder in Figure~\ref{fig:simple-encoder}. Indeed, suppose again that a $Y$
error acts on the second physical qubit in a particular frame of the code
(with the identity acting on all other qubits). Pushing this $Y$ error through
the inverse of the encoder leads to an $X$ operator acting on the ebit and a
$Z$ operator acting on the memory. The $X$ operator is detectable by a Bell
measurement, and the $Z$ operator acting on the memory propagates to a $Z$
operator acting on an ebit in the next frame and then to all ebits and
information qubits in successive frames. These $Z$ operators acting on the
ebits are all detectable by a Bell measurement (whereas they are not detecable
when acting on an ancilla), so that the iterative decoding algorithm can still
correct for errors in the propagation because all these errors trigger syndromes.

\subsection{Recursiveness}

Recursiveness is another desirable property for an EAQ\ convolutional encoder
when it is employed as the inner encoder of a quantum serial turbo
code.\footnote{Recall that the inner encoder is the one closer to the channel,
and the outer encoder is the one farther from the channel.} This property
ensures that a quantum serial turbo code on average has a distance growing
near-linearly with the length of the code, and a proof of this statement
given in Ref.~\cite{OPT08}\ is essentially a direct translation of the classical
proof~\cite{KU98}. In short, an EAQ\ convolutional encoder is
\textit{quasi-recursive} if it transforms every weight-one Pauli operator
$X_{i}$, $Y_{i}$, or $Z_{i}$\ to an infinite-weight Pauli
operator~\cite{PTO09}. It is \textit{recursive} if, in addition to being
quasi-recursive, every element in the logical cosets $C\left(  X_{i}\right)
$, $C\left(  Y_{i}\right)  $, and $C\left(  Z_{i}\right)  $ has infinite
weight. This stringent requirement is necessary due to the coset structure of
EAQ codes.\footnote{This definition thus implies that a recursive encoder is
infinite-depth (it transforms a finite-weight Pauli operator to an
infinite-weight one) \cite{GR06b,WB07}, but the other implication does not
necessarily have to hold.} One might think that quasi-recursiveness is a
sufficient requirement for good performance, but our simulation results in
Section~\ref{sec:sim-results} indicate that recursiveness is indeed necessary
because a turbo code has a significant increase in performance if its inner
encoder is recursive.

It seems like it would be demanding to determine if this condition holds for
every possible input, but we can exploit the state diagram to check it. The
algorithm for checking recursiveness is as follows. First, we define an
\textit{admissable path} to be a path in which its first edge is not part of a
zero physical-weight cycle~\cite{PTO09}. Now, consider any vertex belonging to
a zero-physical weight loop and any admissable path beginning at this vertex
with logical weight one. The encoder is \textit{recursive} if all such paths
do not contain a zero physical-weight loop~\cite{PTO09}. The idea is that a
weight-one logical operator is not sufficient to drive a recursive encoder to a
memory state that is part of a zero-physical weight loop---the minimum weight
of a logical operator that does so is two.

The example encoder in Figure~\ref{fig:simple-encoder} is not recursive. The
only vertex in the state diagram in Figure~\ref{fig:state-diagram} belonging
to a zero physical-weight cycle is the vertex labeled $I$. If the logical
input to the encoder is a $Z$ operator followed by an infinite sequence of
identity operators, the circuit outputs $ZZ$ as the physical output, returns to
the memory state $I$, and then outputs the identity for the rest of time. Thus,
the response to this weight-one input is finite, and the circuit is
non-recursive. Although this example is non-recursive,
Section~\ref{sec:examples} details many examples of EAQ\ convolutional
encoders that are both recursive and non-catastrophic.

\subsection{Distance Spectrum}

We end this section by reviewing the performance measures from
Ref.~\cite{PTO09}, which are also quantizations of the classical
measures~\cite{V71}. The \textit{distance spectrum} $F\left(  w\right)  $\ of
an EAQ\ convolutional encoder is the number of admissable paths beginning and
ending in memory states that are part of a zero physical-weight cycle, where
the physical weight of each admissable path is $w$ and the logical weight is
greater than zero. This performance measure is most similar to the weight
enumerator polynomial of a quantum block code~\cite{PhysRevLett.78.1600},
which helps give an upper bound on the error probability of a non-degenerate
quantum code on a depolarizing channel under maximum likelihood
decoding~\cite{PhysRevA.74.052333}. The distance spectrum incorporates the
translational invariance of a quantum convolutional code and gives an
indication of its performance on a memoryless depolarizing channel.
Appendix~\ref{sec:distance-spectrum} details
a simple method to compute the distance spectrum by using the
state diagram and ideas rooted in Refs.~\cite{V71,M98,book1999conv,McE02}. The
\textit{free distance} of an EAQ\ convolutional encoder is the smallest weight
$w$\ for which $F\left(  w\right)  >0$, and this parameter is one indicator of
the performance of a quantum serial turbo code employing constituent
convolutional encoders. Although one of the applications of our
EAQ\ convolutional encoders are as the inner encoders in a quantum serial
turbo coding scheme, we can also have them as outer encoders in a quantum
serial turbo code and use the free distance to show that its minimum distance
grows near-linearly when combined with a non-catastrophic, recursive inner encoder.

\section{Example Encoders}

\label{sec:examples}Our first example of an EAQ convolutional encoder is the
simplest example that is both recursive and non-catastrophic. It exploits one
memory qubit and one ebit to encode one information qubit per frame. We
discuss this first example in detail, verifying its non-catastrophicity and
recursiveness, and we give a method to compute its distance spectrum.
Appendix~\ref{sec:examples-of-encoders} then
gives tables that detail our other examples. We found all of these examples by
picking encoders uniformly at random from the Clifford group, according to the
algorithm in Section~VI-A.2 of Ref.~\cite{DLT02}.

The seed transformation for our first example is as follows:%
\begin{equation}%
\begin{array}
[c]{ccc}%
Z & I & I\\
I & Z & I\\
I & I & Z\\
X & I & I\\
I & X & I\\
I & I & X
\end{array}
\rightarrow%
\begin{array}
[c]{ccc}%
Z & I & X\\
X & Z & Y\\
X & Y & Z\\
X & X & X\\
Y & I & Y\\
Y & X & Y
\end{array}
, \label{eq:example-encoder}%
\end{equation}
where the first input qubit is the memory qubit, the second input qubit is the
information qubit, the third is Alice's half of the ebit, the first output
qubit is the memory qubit, and the last two outputs are the physical qubits.
We can abbreviate the above encoding by taking the binary representation of
each row at the output and encoding it as a decimal number. For example, the
first row $ZIX$ has the binary representation $\left[  \left.
\begin{array}
[c]{ccc}%
1 & 0 & 0
\end{array}
\right\vert
\begin{array}
[c]{ccc}%
0 & 0 & 1
\end{array}
\right]  $ which is the decimal number 33. Thus, we can specify this encoder
as $\left\{  33,29,30,7,45,47\right\}  $.\footnote{Note that this convention
is different from that of Poulin \textit{et al}.~in Ref.~\cite{PTO09}.}

The seed transformation in (\ref{eq:example-encoder}) leads to the state
diagram of Figure~\ref{fig:1st-example}, by exploiting (\ref{eq:state-diagram}%
). We can readily check that the encoder is non-catastrophic and recursive by
inspecting Figure~\ref{fig:1st-example}. The only cycle with zero physical
weight is the self-loop at the identity memory state with zero logical weight.
The encoder is thus non-catastrophic. To verify recursiveness, note again that
the only vertex belonging to a zero physical-weight cycle is the self-loop at
the identity memory state. We now consider all weight-one admissable paths
that begin in this state. If we input a logical $X$, we follow the edge
$\left(  X,IY\right)  $ to the $Y$ memory state. Inputting the identity
operator for the rest of time keeps us in the self-loop at memory state $Y$
while still outputting non-zero physical weight operators. We can then check
that inputting a $Z$ logical operator followed by identities keeps us in the
self-loop at the $X$ memory state while still outputting non-zero physical
weight operators, and inputting a $Y$ logical operator followed by identities
keeps us in the self-loop at the $Z$ memory state. The encoder is thus recursive.
Appendix~\ref{sec:examples-of-encoders} lists many more examples of
entanglement-assisted quantum convolutional encoders that are both recursive
and non-catastrophic.

\begin{figure}
[ptb]
\begin{center}
\includegraphics[
natheight=1.753800in,
natwidth=1.233200in,
height=2.8816in,
width=2.0358in
]%
{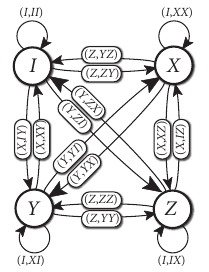}%
\caption{The state diagram corresponding to the seed transformation in
(\ref{eq:example-encoder}), acting on one memory qubit, one information qubit,
one ebit, and producing two physical qubits. The above state diagram allows us
to check that the encoder is both recursive and non-catastrophic.}%
\label{fig:1st-example}%
\end{center}
\end{figure}
%

\section{EAQ Turbo Codes}

\label{sec:EAQ-turbo}We comment briefly on the construction and minimum
distance of a quantum serial turbo code that employs our example encoders as a
constituent encoder (the next section details our simulation results with
different encoders). The construction of an EAQ serial turbo code is the same
as that in Ref.~\cite{PTO09}\ (see Figure~10 there), with the exception that
we assume that Alice and Bob share entanglement in the form of ebits before
encoding begins. Alice first encodes her stream of information qubits with the
outer encoder, performs a quantum interleaver on all of the qubits, and then
encodes the resulting stream with the inner encoder. The quantum communication
rate of the resulting EAQ\ turbo code is $k^{\text{Out}}/n^{\text{In}%
}=(k^{\text{Out}}/n^{\text{Out}})(k^{\text{In}}/n^{\text{In}})$ where
$k^{\text{Out}}$ is the number of information qubits encoded by the outer
encoder, $n^{\text{Out}}$ is the number of physical qubits output from the
outer encoder, and a similar convention holds for $k^{\text{In}}$,
$n^{\text{In}}$, and the inner encoder. In order for the qubits to match up
properly, $n^{\text{Out}}$ must be equal to $k^{\text{In}}$. The entanglement
consumption rate of the code is $(c^{\text{Out}}+c^{\text{In}})/n^{\text{In}}$
where $c^{\text{Out}}$ and $c^{\text{In}}$ are the total number of ebits
consumed by the outer and inner encoder, respectively.

Perhaps the best combination for an EAQ\ serial turbo code is to choose the
inner quantum convolutional encoder to be a recursive, non-catastrophic EAQ
convolutional encoder and the outer quantum convolutional encoder to be a
non-recursive, non-catastrophic standard quantum convolutional encoder. This
combination reduces entanglement consumption, ensures good iterative decoding
performance, and ensures that the minimum distance of the quantum serial turbo
code grows as $N^{\left(  d^{\ast}-2\right)  /d^{\ast}}$ where $N$ is the
length of the code and $d^{\ast}$ is the free distance of the outer quantum
convolutional encoder. A proof of this last statement given in Ref.~\cite{OPT08} is
essentially identical to the classical proof~\cite{KU98}.
Section~\ref{sec:noise-ebits}\ shows that this combination also performs well
in practice if noise occurs on the ebits. Additionally, choosing the outer
quantum convolutional encoder to encode a highly degenerate code may increase
the number of errors that the quantum turbo code can correct, in a vein
similar to the results in Ref.~\cite{PhysRevA.57.830} (though we have not yet
fully investigated this possibility). Appendix~\ref{sec:examples-of-encoders} lists
many examples of entanglement-assisted quantum turbo codes and discusses their average
minimum distance scaling.

\section{Incorporating Extrinsic Information into the
Quantum Turbo Decoding Algorithm}
\label{sec:turbo-decode-mod}

Poulin {\it et al.} proposed an iterative decoding
algorithm for quantum turbo codes in Ref.~\cite{PTO09}.
Their algorithm is based on the exchange of
\textit{a posteriori} information between the 
constituent quantum convolutional decoders of a quantum turbo code.
Since the decoders pass along \textit{a posteriori} information,
successive iterations of the constituent decoders
are dependent on one another, and this gives rise to a detrimental positive
feedback effect, which prevents the decoding algorithm there
from achieving the desired gains usually observed in iterative decoding.

To avoid the aforementioned situation, it is necessary
to ensure that the \textit{a priori} information directly
related to a given information qubit is not reused in the other
constituent decoder~\cite{BGT93,tc_book_hanzo}. Similar
to the approach employed in classical turbo decoding~\cite{BGT93,SISO97},
this can be achieved by having one decoder remove \textit{a priori}
information from the \textit{a posteriori} information before feeding
it to the other decoder. More explicitly,
the iterative decoding procedure
should exchange only ``extrinsic'' information 
 which is unknown and new to the other decoder.

To see how this works, let us consider a four-port Soft-Input Soft-Output (SISO)
decoder~\cite{SISO97,vlc_book_hanzo} that generates
soft output information pertaining to a logical error
$L$ and physical error $P$. As shown in Figure~\ref{fig:SISO}, a
SISO decoder exploits an \textit{A Posteriori} Probability
(APP) module that accepts the \textit{a priori} information
$\mathbf{P}^a(L)$ and $\mathbf{P}^a(P)$ as input and outputs
the \textit{a posteriori} information $\mathbf{P}^o(L)$
and $\mathbf{P}^o(P)$. The corresponding
extrinsic probabilities $\mathbf{P}^e(L^j_i)$ and
$\mathbf{P}^e(P^j_i)$ for the $j^{\text{th}}$ qubit at time
instant $i$ are then obtained by discarding the \textit{a priori}
information from the \textit{a posteriori} information
as follows~\cite{SISO97,vlc_book_hanzo}:
\begin{align}
 \mathbf{P}^e(L^j_i) = N_{L^j}\frac{\mathbf{P}^o(L^j_i)}{\mathbf{P}^a(L^j_i)}, \nonumber \\
 \mathbf{P}^e(P^j_i) = N_{P^j}\frac{\mathbf{P}^o(P^j_i)}{\mathbf{P}^a(P^j_i)},
\label{eq:ext}
\end{align}
where $N_{L^j}$ and $N_{P^j}$ are normalization
factors, which ensure that $\sum_{L^j}\mathbf{P}^e(L^j_i) = 1$
and $\sum_{P^j}\mathbf{P}^e(P^j_i) = 1$, respectively.
Furthermore, to avoid any numerical instabilities and to reduce
the computational complexity, log-domain arithmetics
are conventionally
employed, which convert the multiplicative operations
of (\ref{eq:ext}) to addition, as given below~\cite{vlc_book_hanzo}:
\begin{align}
 \ln[\mathbf{P}^e(L^j_i)] &= \ln[\mathbf{P}^o(L^j_i)] - \ln[\mathbf{P}^a(L^j_i)], \nonumber \\
 \ln[\mathbf{P}^e(P^j_i)] &= \ln[\mathbf{P}^o(P^j_i)] - \ln[\mathbf{P}^a(P^j_i)].
\label{eq:log_ext}
\end{align}
Therefore, the inputs and outputs of a SISO decoder
are the logarithmic \textit{a priori} and extrinsic probabilities, respectively.
\begin{figure}
[ptb]
\begin{center}
\includegraphics[scale = 0.65]%
{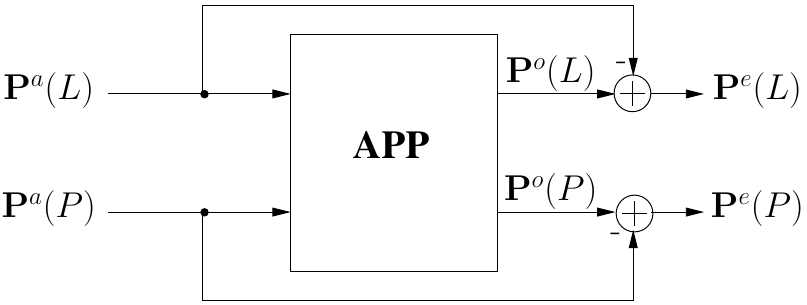}%
\caption{A SISO decoder processes soft \textit{a piori}
information to yield soft extrinsic information which is useful
for iterative decoding.
The quantities $\mathbf{P}^a(\cdot)$, $\mathbf{P}^e(\cdot)$ and $\mathbf{P}^o(\cdot)$
denote the {\it a priori}, extrinsic and {\it a posteriori} probabilities, while
$L$ and $P$ denote the logical and physical errors.}%
\label{fig:SISO}%
\end{center}
\end{figure}

Based on the above discussion, we have modified the turbo
decoding algorithm of Ref.~\cite{PTO09} so that the constituent decoders
pass along only
extrinsic information to each other. See Figure~\ref{fig:turbo_dec}
for a depiction of the modified quantum turbo decoding
algorithm \cite{BNH13}. In this figure, $A(y)$ and $E(y)$
denote the logarithmic \textit{a priori} and extrinsic probabilities 
of $y$, where $y \in \{L_1, L_2, P_1, P_2\}$. Analogous
to Ref.~\cite{PTO09}, the inner SISO decoder of
Figure~\ref{fig:turbo_dec} exploits the 
physical noise model $A(P_1)$, syndrome $S^x_1$ and \textit{a priori}
information $A(L_1)$, the last of which is initialized to be equiprobable.
However, it outputs the extrinsic information 
$E(L_1)$, rather than the \textit{a posteriori} information
as in Ref.~\cite{PTO09}. The extrinsic information
$E(L_1)$ is then interleaved to serve as the \textit{a priori}
information $A(P_2)$ for the outer SISO decoder. The two decoders
thereby engage in iterative decodings, which continue
until either the \textit{a posteriori} probability
$E(L_2)$ converges
to a definite solution, or a prespecified
maximum number of iterations is reached.\footnote{Since $A(L_2) = 0$,
the \textit{a posteriori} probability of $L_2$ is the same as
the extrinsic information, according to (\ref{eq:log_ext}).}
\begin{figure}
[ptb]
\begin{center}
\includegraphics[scale = 0.55]%
{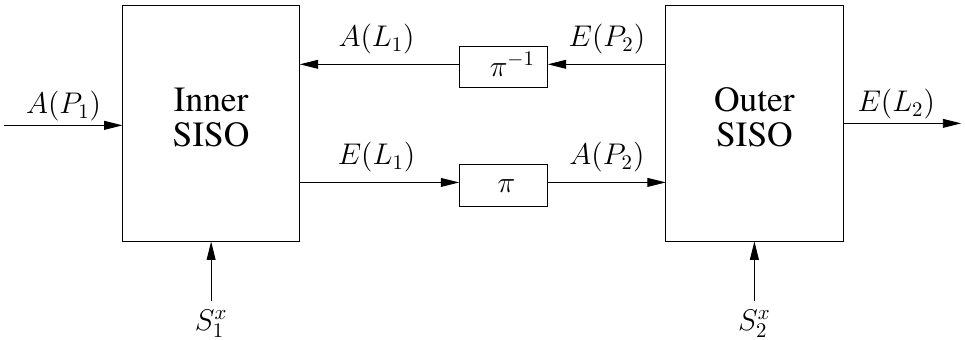}%
\caption{Schematic of the quantum turbo decoder algorithm 
that exploits extrinsic information rather than {\it a posteriori}
information as in Ref.~\cite{PTO09}. In the above, $A(\cdot)$ and $E(\cdot)$
denote the logarithmic {\it a priori} and
extrinsic probabilities, $L_i$ and $P_i$ are the logical
and physical errors, while $S^x_i$ represents the $X$-syndrome
for the $i^{\text{th}}$ decoder. The generalization of this algorithm
to the entanglement-assisted case is straightforward.}%
\label{fig:turbo_dec}%
\end{center}
\end{figure}

\section{Simulation Results}

\label{sec:sim-results}We performed several simulations of EAQ turbo codes and
detail the results in this section. This section begins with a description of
the parameters of the constituent quantum convolutional encoders. We then
describe how the simulation was run, and we finally discuss and interpret the
simulation results.

Our simulation results presented here are certainly not intended
to be an exhaustive comparison of codes, but they rather serve to illustrate
the enhancement in performance from the modified algorithm described in
Section~\ref{sec:turbo-decode-mod} and they
constitute an exploration
of the effect of adding entanglement assistance to the encoders of Poulin \textit{et
al}.~\cite{PTO09}---our original intent was to exploit their encoders, but the
lack of a clear exposition of their decimal representation convention has
excluded us from doing so. So we instead randomly generated and filtered
encoders that have comparable attributes to their encoders. The first encoder,
dubbed \textquotedblleft PTO1R,\textquotedblright\ acts on three input memory
qubits, one information qubit, and two ancillas to produce three output memory
qubits and three physical qubits. Its decimal representation is%
\[%
\begin{array}
[c]{c}%
\{1355,2847,558,2107,3330,739,\\
\ \ \ 2009,286,473,1669,1979,189\}
\end{array}
,
\]
and a truncated distance spectrum polynomial
(see Appendix~\ref{sec:distance-spectrum}) for it is%
\begin{multline*}
11x^{5}+47x^{6}+253x^{7}+1187x^{8}+\\
6024x^{9}+30529x^{10}+153051x^{11}+771650x^{12}.
\end{multline*}
This encoder is non-catastrophic and quasi-recursive, and its parameters and
distance spectrum are comparable to those of Poulin \textit{et al}.'s first
encoder (though they did not comment on whether theirs is
quasi-recursive)~\cite{PTO09}. Replacing the two ancilla qubits in the encoder
with two ebits gives an improvement in its truncated distance spectrum
polynomial:%
\begin{multline*}
2x^{9}+x^{10}+5x^{11}+8x^{12}+11x^{13}+25x^{14}+\\
56x^{15}+102x^{16}+217x^{17}+387x^{18}+787x^{19}.
\end{multline*}
The encoder also becomes recursive after this replacement. Let
\textquotedblleft PTO1REA\textquotedblright\ denote the EA\ version of this encoder.

Our second encoder, dubbed \textquotedblleft PTO3R,\textquotedblright\ acts on
four input memory qubits, one information qubit, and one ancilla qubit to
produce four output memory qubits and two physical qubits. Its decimal
representation is%
\[%
\begin{array}
[c]{c}%
\{3683,3556,2872,2211,3561,3534,\\
\ \ \ 729,3136,743,2643,1330,1656\}
\end{array}
,
\]
and a truncated distance spectrum polynomial for it is%
\begin{multline*}
12x^{5}+93x^{6}+600x^{7}+4320x^{8}+31098x^{9}+\\
224014x^{10}+1604435x^{11}+11469935x^{12}.
\end{multline*}
The encoder is non-catastrophic and quasi-recursive with its parameters and
distance spectrum comparable to those of Poulin \textit{et al}.'s third
encoder. Replacing the ancilla with an ebit gives an improvement to the
truncated distance spectrum polynomial:%
\begin{multline*}
x^{6}+3x^{7}+7x^{8}+29x^{9}+88x^{10}+237x^{11}+\\
716x^{12}+2166x^{13}+6245x^{14}+18696x^{15}+\\
55889x^{16}+165971x^{17}+492805x^{18}+1465529x^{19},
\end{multline*}
and the encoder becomes recursive after this replacement. Let
\textquotedblleft PTO3REA\textquotedblright\ denote the EA version of the
\textquotedblleft PTO3\textquotedblright\ encoder.

Our turbo codes consist of interleaved serial concatenation of the above
encoders with themselves, and we varied the auxiliary resources of the
encoders to be ebits, ancillas, or both. Concatenating PTO1R with itself leads
to a rate 1/9 quantum turbo code, and concatenating PTO3R with itself leads to
a rate 1/4 turbo code.%
\begin{figure}
[ptb]
\begin{center}
\includegraphics[
natheight=3.846700in,
natwidth=5.440500in,
height=2.4422in,
width=3.4411in
]%
{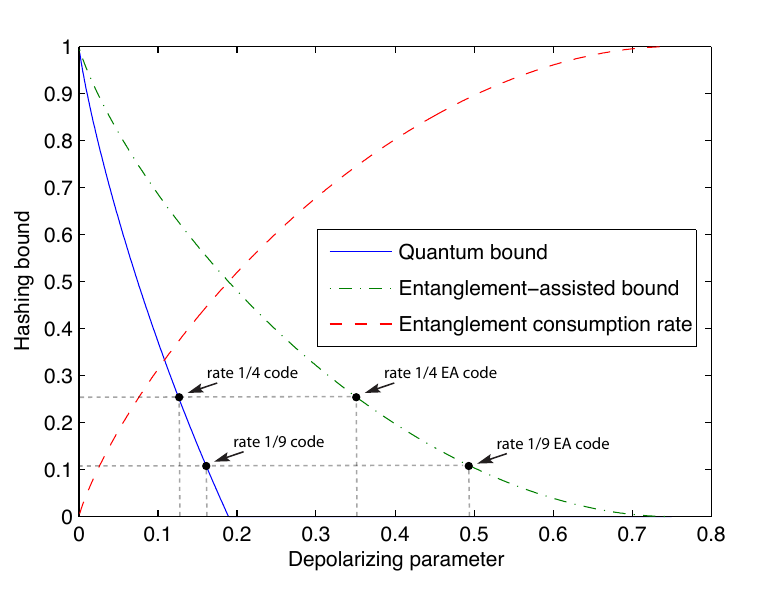}%
\caption{The figure plots the hashing bound for quantum codes (solid blue
curve), the entanglement-assisted hashing bound for EAQ codes (the dash-dotted
green curve), and the entanglement consumption rate of the father protocol
that meets the entanglement-assisted hashing bound (dashed red curve). The
figure plots these different bounds as a function of the depolarizing
parameter $p$. Observe that the noise limit for a rate 1/4 quantum code is
$\approx0.12689$, the noise limit for a rate 1/9 quantum code is
$\approx0.16028$, the noise limit for a rate 1/4 EAQ code is $\approx0.35454$,
and the noise limit for a rate 1/9 EAQ code is $\approx0.49088$.
We should compare our simulation results with these noise limits to observe how close
our code constructions come to achieving the quantum capacity (though note that the
unassisted hashing bound is only a lower bound on the unassisted quantum
capacity \cite{SS96,PhysRevA.57.830}).}%
\label{fig:hashing}%
\end{center}
\end{figure}

We simulated the performance of these EAQ turbo codes on a memoryless
depolarizing channel with parameter $p$. The action of the channel on a single
qubit density operator $\rho$\ is as follows:%
\[
\rho\rightarrow\left(  1-p\right)  \rho+\frac{p}{3}\left(  X\rho X+Y\rho
Y+Z\rho Z\right)  .
\]
The benchmarks for the performance of any standard or entanglement-assisted
quantum code on a depolarizing channel are the hashing
bounds~\cite{PhysRevA.54.3824,PhysRevLett.83.3081}. The quantum hashing bound
determines the rate at which a random quantum code can operate reliably for a
particular depolarizing parameter $p$, but note that it does not give the
ultimate capacity limit because degeneracy can improve
performance~\cite{PhysRevA.57.830,PhysRevLett.98.030501,PhysRevA.78.062335,PhysRevA.77.010301}%
. On the other hand, the entanglement-assisted hashing bound does give the
ultimate capacity at which an EAQ code can operate reliably for a particular
depolarizing parameter $p$, whenever an unbounded amount of entanglement is
available~\cite{PhysRevLett.83.3081,PhysRevA.66.052313}. The hashing bound for
quantum communication is%
\[
1-\left[  H_{2}\left(  p\right)  +p\log_{2}\left(  3\right)  \right]  ,
\]
and the hashing bound for an entanglement-assisted quantum code is%
\begin{equation}
1-\frac{1}{2}\left[  H_{2}\left(  p\right)  +p\log_{2}\left(  3\right)
\right]  . \label{eq:EA-hashing-bound}
\end{equation}
The father protocol is a quantum communication protocol that achieves the
entanglement-assisted quantum
capacity~\cite{PhysRevLett.93.230504,arx2005dev,arx2006anura}, while
attempting to minimize its entanglement consumption rate. Its entanglement
consumption rate is provably optimal for certain channels such as the
dephasing channel~\cite{HW08a,HW09,WH10a}, but it is not necessarily optimal
for the depolarizing channel. The entanglement consumption rate of the father
protocol on a depolarizing channel is%
\[
\frac{1}{2}\left[  H_{2}\left(  p\right)  +p\log_{2}\left(  3\right)  \right]
.
\]
This entanglement consumption rate implies that the only resources involved in
an encoded transmission of the father protocol are information qubits and
ebits because the sum of its quantum communication rate and its entanglement
consumption rate is equal to one. Figure~\ref{fig:hashing}\ plots these
different hashing bounds and depicts the location of these bounds for our
different quantum turbo codes.

One can also consider a ``hashing region'' of rates that are achievable
for entanglement-assisted quantum communication
\cite{arx2005dev,PhysRevLett.93.230504,LBW10}---this region is the more relevant
benchmark for an entanglement-assisted code that does not consume the maximal
amount of entanglement. For the case of a depolarizing channel, this region
consists of all $Q$ and $E$ that satisfy the following bounds:
\begin{align}
Q & \leq 1-\left[  H_{2}\left(  p\right)  +p\log_{2}\left(  3\right)  \right] + E
\label{eq:boundary-hashing-region-1}\\
Q & \leq 1-\frac{1}{2}\left[  H_{2}\left(  p\right)  +p\log_{2}\left(  3\right)
\right],
\end{align}
where $Q$ is the quantum communication rate and $E$ is the entanglement consumption
rate. The intersection of these two boundary lines occurs when
$E = \frac{1}{2}\left[  H_{2}\left(  p\right)  +p\log_{2}\left(  3\right)  \right]$,
the entanglement consumption rate of the father protocol.
When an entanglement-assisted quantum code does not consume the maximal amount of
entanglement possible (but rather at some rate
$E \leq \frac{1}{2}\left[  H_{2}\left(  p\right)  +p\log_{2}\left(  3\right)  \right]$),
we should compare its performance with the boundary in (\ref{eq:boundary-hashing-region-1}).
\begin{figure*}
[ptb]
\begin{center}
\includegraphics[
height=5.93in,
]%
{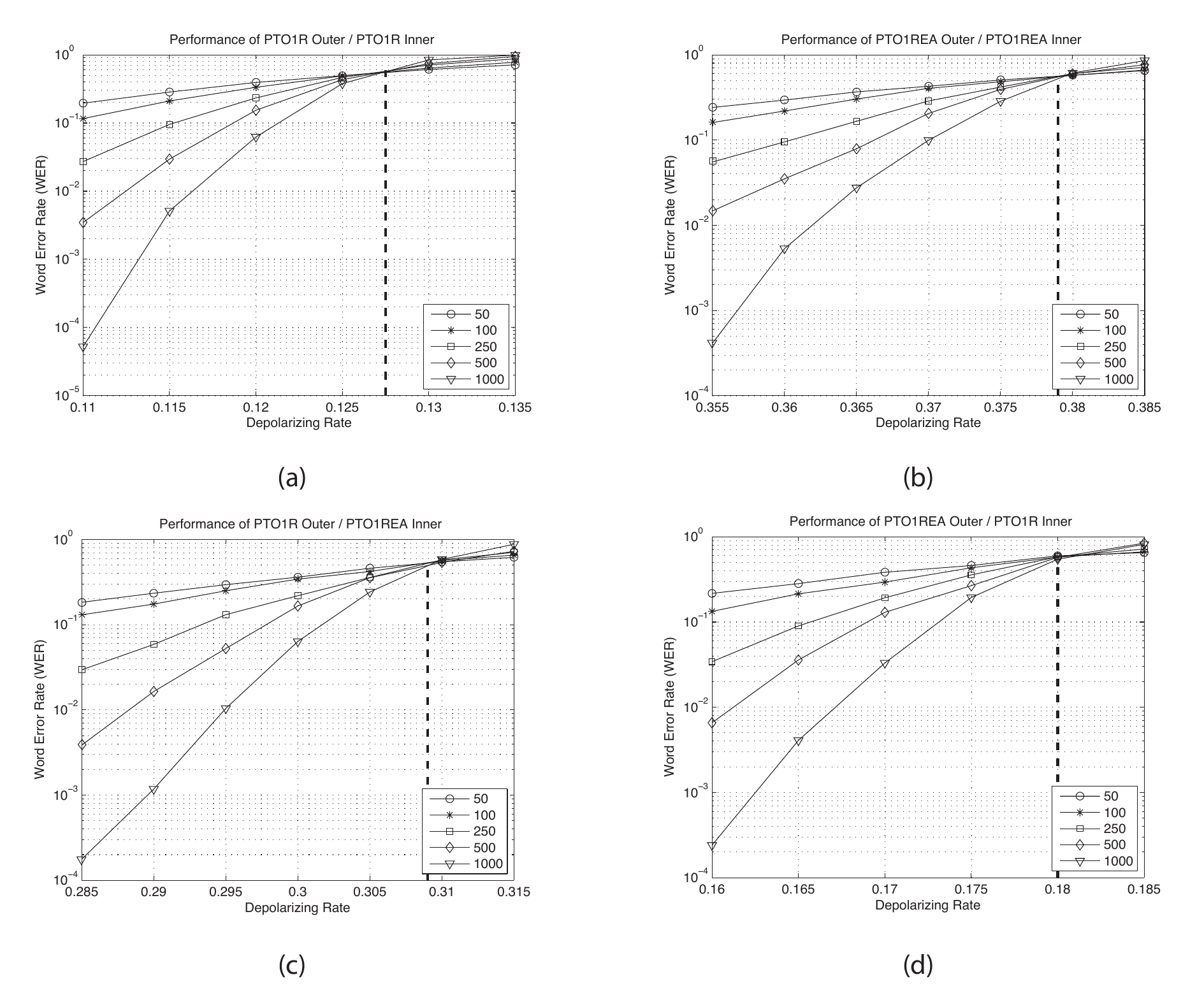}%
\caption{The figure plots the results of four different simulations that
exploit the PTO1R and PTO1REA\ encoders. The resulting EAQ\ turbo codes have
rate 1/9 and varying entanglement consumption rates. Each plot demonstrates
the existence of a threshold, where the code performance increases whenever
the channel noise rate is below the threshold. The vertical dotted line in
each plot indicates the approximate location of the threshold.}%
\label{fig:PTO1-performance}%
\end{center}
\end{figure*}

\begin{figure*}
[ptb]
\begin{center}
\includegraphics[
height=6.0009in,
]%
{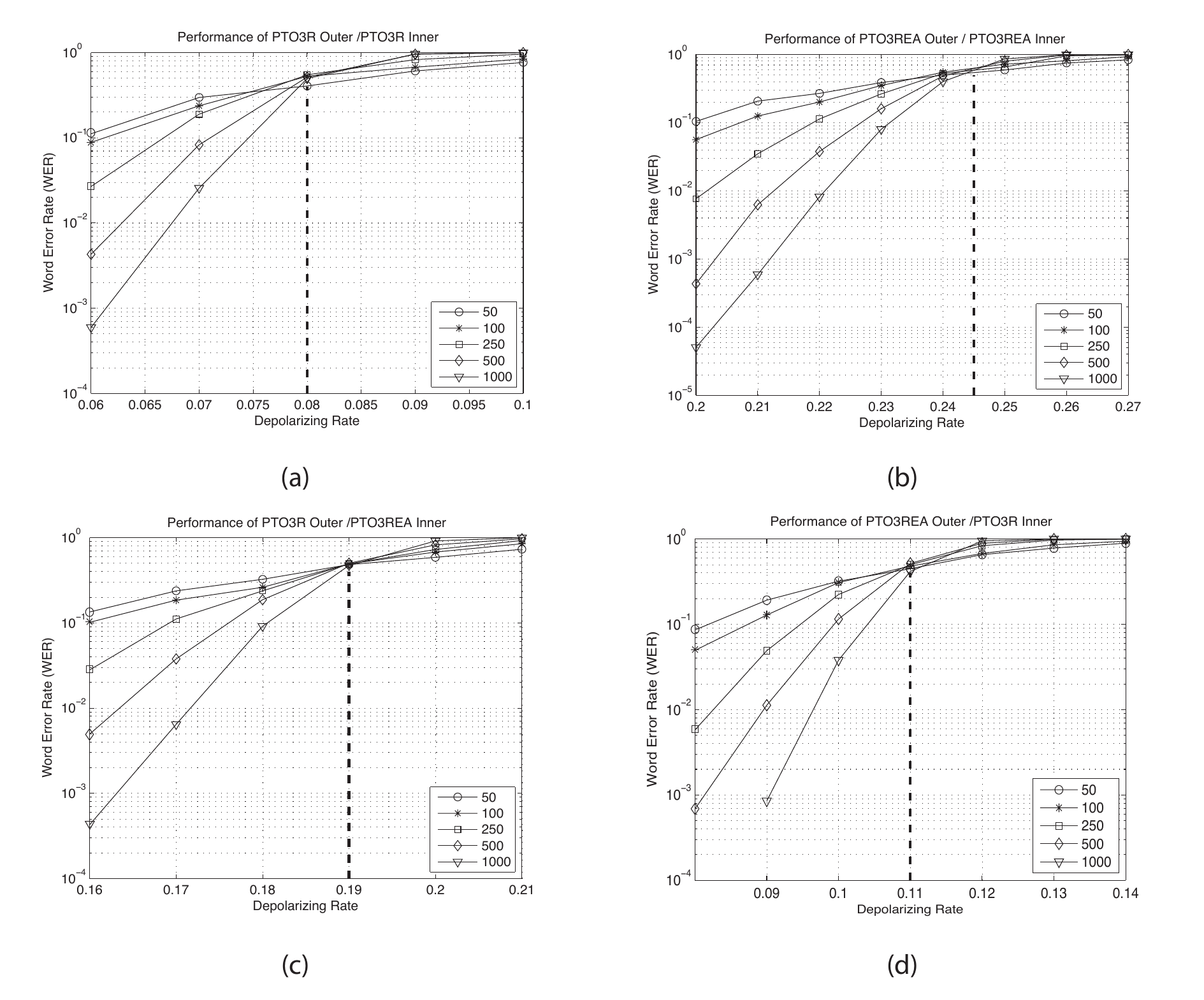}%
\caption{The figure plots the results of four different simulations that
exploit the PTO3R and PTO3REA\ encoders. The resulting EAQ\ turbo codes have
rate 1/4 and varying entanglement consumption rates. Each plot demonstrates
the existence of a threshold, where the code performance increases whenever
the channel noise rate is below the threshold. The vertical dotted line in
each plot indicates the approximate location of the threshold.}%
\label{fig:PTO3-performance}%
\end{center}
\end{figure*}

We performed Monte Carlo simulations to determine the performance of our
example EAQ turbo codes when decoded with an iterative decoding algorithm (the
iterative decoding algorithm is as described in
Section~\ref{sec:turbo-decode-mod}, with
the exception that it decodes EAQ turbo codes). We developed a Matlab computer
program for these simulations~\cite{W10}.

A single run of a simulation selects a quantum turbo code with a particular
number of logical qubits and a random choice of interleaver, and it then
generates a Pauli error randomly according to a depolarizing channel with
parameter $p$. This Pauli error leads to an error syndrome, and the syndrome
and channel model act as inputs to the iterative decoding algorithm. The
iterative decoding algorithm terminates if the hard decision on a recovery
operation is the same decision from the previous iteration, or it terminates
after a maximum number of iterations (though we never observed the number of
iterations until convergence exceeding eight). One run of the simulation
declares a failure if the estimated recovery operation is different from the
correct recovery operation. The ratio of simulation failures to the total
number of simulation runs is the \textit{word error rate} (WER). In the cases
in which errors occur more rarely, we ran every
configuration (choice of code, depolarizing parameter, and number of logical
qubits) until we observed at least 100 failures---this number gaurantees a
reasonable statistical confidence in the results of the simulations. Of course,
in the cases where errors occur more frequently, we observed thousands of errors.

Our first simulation involved the serial concatenation of the unassisted PTO1R
encoder with itself, and Figure~\ref{fig:PTO1-performance}(a) displays the
results. The performance significantly exceeds that of the first encoder of Poulin
\textit{et al}.~(see Figure~12\ of Ref.~\cite{PTO09}) and there is
an improved separation between the curves for increasing blocklength, 
when comparing Figure~\ref{fig:PTO1-performance}(a) to 
Figure~12\ of Ref.~\cite{PTO09}. A noticeable feature of
Figure~\ref{fig:PTO1-performance}(a) is the existence of a pseudothreshold,
such that increasing the number of encoded qubits of the turbo code decreases
the WER\ for all depolarizing noise rates below the pseudothreshold. The
pseudothreshold is not a true threshold because this particular turbo code has
a bounded minimum distance,\footnote{As discussed in Section~4.4 of Ref.~\cite{PTO09},
the minimum distance of a quantum turbo code is upper bounded by the weight-one
minimum distance of the inner encoder times the free distance of the outer encoder. 
Whenever the inner encoder is non-recursive, the weight-one minimum distance is bounded by some
constant (not growing with the blocklength), implying a constant
bound on the minimum distance of the turbo code.} and at some point, the WER\ should begin
increasing if we continue to increase the number of encoded
qubits~\cite{PTO09}. The pseudothreshold occurs at a depolarizing noise rate
$\approx 0.1275$, and this pseudothreshold is within $0.994$\ dB of the
$0.16028$\ noise limit for a rate $1/9$ code.

Our second simulation tested the serial concatenation of the PTO1REA encoder
with itself, and Figure~\ref{fig:PTO1-performance}(b) displays the results.
This turbo code uses the maximal amount of entanglement at a rate of $8/9$ and
thus is an instance of the so-called ``father'' protocol. Entanglement assistance gives the
turbo code a dramatic increase in performance, in the sense that it can withstand
far higher depolarizing noise rates than the unassisted turbo code. The
threshold occurs at $\approx0.379$---we call
this a threshold rather than a pseudothreshold
because we expect that the WER\ should
continue to decrease as we increase the
number of encoded qubits, but we should clarify that we have
not proven that this should hold. However,
the EXIT chart analysis of Ref.~\cite{BNH13} suggests that this is
a true threshold. This
threshold is $4.73$\ dB beyond the pseudothreshold of the unassisted turbo
code, and it is within $1.12$~dB of the $0.49088$~noise limit
given by the
EA\ hashing bound in (\ref{eq:EA-hashing-bound})
for
a qubit rate $1/9$ and ebit rate $8/9$ code.
This code construction is operating in a noise regime
in which standard quantum codes are simply not able to operate (compare with
the results in Refs.~\cite{MMM04,HI07,COT07,PTO09,HBD09,TL10,KHIS10,FCVBT10}).

Our third simulation tested the serial concatenation of the PTO1REA inner
encoder with the PTO1R outer encoder, and this EAQ\ turbo code has an
entanglement consumption rate of $6/9$. The benchmark for comparison is given
by (\ref{eq:boundary-hashing-region-1}), so that for a
code with qubit rate $1/9$ and ebit rate $6/9$, its noise limit is found
by solving for the $p$ for which
$$
1/9 - 6/9 = -5/9 = 1-\left[  H_{2}\left(  p\right)  +p\log_{2}\left(  3\right)  \right],
$$
which in this case is $p \approx 0.3779$.
The inner encoder is recursive, but
the outer encoder's free distance is not as high as that in the previous
simulation.
Figure~\ref{fig:PTO1-performance}(c) displays the
results. The threshold occurs
approximately at a depolarizing noise rate of $0.309$, which is within
$0.887$~dB of the previous threshold and within $0.87$~dB of the $0.3779$%
~noise limit given above.

Our final simulation tested the serial concatenation of the PTO1R inner
encoder with the PTO1REA outer encoder, and this EAQ turbo code has an
entanglement consumption rate of $2/9$. Again, the benchmark for comparison is given
by (\ref{eq:boundary-hashing-region-1}), so that for a
code with qubit rate $1/9$ and ebit rate $2/9$, its noise limit is found
by solving for the $p$ for which
$$
1/9 - 2/9 = -1/9 = 1-\left[  H_{2}\left(  p\right)  +p\log_{2}\left(  3\right)  \right],
$$
which in this case is $p \approx 0.2209$. The inner encoder is quasi-recursive,
and the outer encoder has a significantly higher free distance than in the
previous simulation because it has entanglement assistance. Though, one can place
a constant upper bound on the minimum distance of these turbo codes because
the inner encoder is not recursive. This simulation
provides a good test to determine the effectiveness of an inner encoder that
is quasi-recursive. That is, one might think that quasi-recursiveness of the
inner encoder combined with an outer encoder with high free distance would be
sufficient to produce a turbo code with good performance (the thought is that
this would explain the good performance in the first simulation), but
Figure~\ref{fig:PTO1-performance}(d) suggests that this intuition does not
hold. The pseudothreshold occurs at approximately $0.18$. This
pseudothreshold is only $1.50$~dB higher than the threshold for the unassisted
code, and it is $0.889$~dB away from the noise limit of $0.2209$ given above.

We conducted similar simulations with the PTO3R and PTO3REA encoders, and
Figures~\ref{fig:PTO3-performance}(a-d)\ display the results. The entanglement
consumption rates in Figures~\ref{fig:PTO3-performance}(a-d) are $0$, $3/4$,
$2/4$, and $1/4$, respectively. The results are somewhat similar to the previous
simulations with the difference that the noise limits and thresholds are lower
because these codes have higher quantum data transmission rates. The thresholds
occur at approximately $0.08$, $0.245$, $0.19$, and $0.11$
in Figures~\ref{fig:PTO3-performance}(a-d), respectively, and the hashing limits
for these codes are approximately $0.1268$,
$0.3545$, $0.2636$, and $0.1893$, respectively. Thus, these codes are within
2~dB, 1.6~dB, 1.4~dB, and 2.35~dB of their hashing limits, respectively.
Since the threshold of the code in
Figure~\ref{fig:PTO3-performance}(c) occurs at $0.19$, so that it is
performing closest to its hashing limit (1.4 dB away), it appears
that this EAQ\ turbo code is making judicious use of the available
entanglement by placing the two ebits in the inner encoder. We would have to
conduct further simulations to determine if placing one ebit in the outer
encoder and one in the inner encoder would do better, but our suspicion is
that the aforementioned use of entanglement is better because the inner
encoder is recursive under this choice.

\subsection{Noise on Ebits}

\label{sec:noise-ebits}We conducted another set of simulations to determine
how noise on Bob's half of the ebits affects the performance of a code. This
possibility has long been one of the important practical questions concerning
the entanglement-assisted paradigm, and some researchers have provided partial
answers~\cite{prep2007shaw,WF09,WH10,PhysRevA.79.042342,PhysRevA.86.032319}. We briefly review
some of these contributions. Shaw \textit{et al}.~first observed that the
$\left[  \left[  7,1,3\right]  \right]  $ Steane code is equivalent to a
$\left[  \left[  6,1,3;1\right]  \right]  $ entanglement-assisted code that
can also correct a single error on Bob's half of the ebit. This observation
goes further:\ any standard, non-degenerate $\left[  \left[  n,k,d\right]
\right]  $\ quantum code is equivalent to an $\left[  \left[
n-d+1,k,d;d-1\right]  \right]  $ entanglement-assisted code that can correct
any $\left\lfloor \left(  d-1\right)  /2\right\rfloor $ errors on Alice and
Bob's qubits. This result holds because tracing over any $n-d+1$ qubits in the
original standard code gives a maximally mixed state on $d-1$ qubits, and the
purification of these qubits are $n-d+1$ encoded halves of ebits that Alice
possesses~\cite{P99}. Lai and Brun
have studied this observation in more detail, by conducting simulations of such
entanglement-assisted codes, and they have also studied the case in which a standard stabilizer code
is used to protect the ebits \cite{PhysRevA.86.032319}.
Wilde and Fattal observed that entanglement-assisted
codes correcting for errors on Bob's side slightly improve the threshold for
quantum computation if ebit errors occur less frequently than gate
errors~\cite{WF09}. Hsieh and Wilde studied this question in the
Shannon-theoretic context and determined an expression for capacity when
channel errors and entanglement errors occur~\cite{WH10}. As a side note, Lai
and Brun have also looked for codes attempting to achieve the opposite
goal~\cite{LB10}---their codes try to maximize the number of channel errors
that can be corrected while minimizing the correction on Bob's half of the ebits.

\begin{figure}
[ptb]
\begin{center}
\includegraphics[
width=3.7411in
]%
{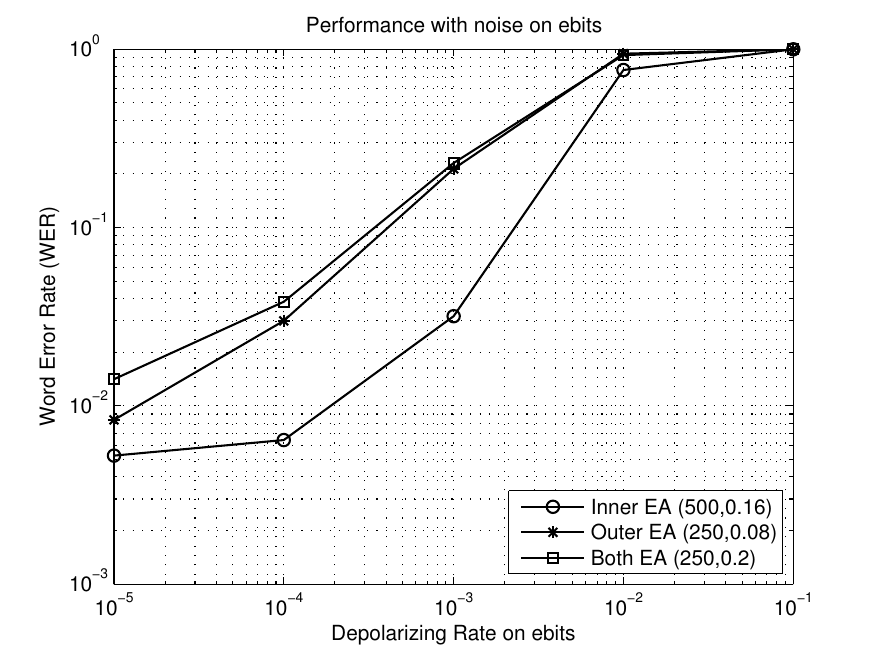}%
\caption{The figure plots the performance of different inner and outer encoder
combinations as a function of noise rate on Bob's half of the ebits. The number of
logical qubits and channel noise for each combination is indicated in the figure legend. It is
surprising that the combination PTO3REA inner / PTO3R outer can withstand more
ebit noise than the combination PTO3R inner / PTO3REA outer even though it
uses more ebits than this latter combination.}%
\label{fig:noise-on-ebits}%
\end{center}
\end{figure}
Figure~\ref{fig:noise-on-ebits}\ plots the results of our simulations that
allow for noise on Bob's half of the ebits. We performed three different types
of simulations:\ the first was with the PTO3REA inner / PTO3R outer
combination, the second with the PTO3R inner / PTO3REA outer combination, and
the third with the PTO3REA inner / PTO3REA outer combination. For each
simulation, we kept the channel noise rates fixed at $0.16$, $0.08$, and
$0.2$, respectively, because the codes already performed reasonably well at
these noise rates, and our goal was to understand the effect of ebit noise on
code performance. The codes all performed about the same as they did without
ebit noise if we set the ebit noise rate at $10^{-5}$. Increasing the ebit
noise rate an order of magnitude to $10^{-4}$ has the least effect on the
PTO3REA inner / PTO3R outer combination. Increasing it further to $10^{-3}$
deteriorates the performance of the other combinations while still having the
least effect on the PTO3REA inner / PTO3R outer combination. This result is
surprising, considering that this combination has an entanglement consumption
rate of $2/4$ while the entanglement consumption rate of the PTO3R inner /
PTO3REA outer combination is smaller at $1/4$. The result suggests that it
would be wiser to place ebits in the inner encoder rather than in the outer
encoder when these codes operate in practice.

Figure~\ref{fig:noise-on-ebits-1000} seems to confirm this suggestion. Increasing the number
of logical qubits to 1000 for each combination shows an increased performance for the combination
PTO3REA inner / PTO3R outer if the ebit noise level is not
too high, while the other two combinations perform worse. It is surprising that
the combination PTO3REA inner / PTO3R outer performs better---increasing 
the number of logical qubits in turn increases the number of ebits, having more ebits should translate
into more noise on the syndromes, and this then should affect performance. However, this
particular combination seems to exhibit some amount of robustness against ebit noise if it is not too high.

\begin{figure}
[ptb]
\begin{center}
\includegraphics[
width=3.7411in
]%
{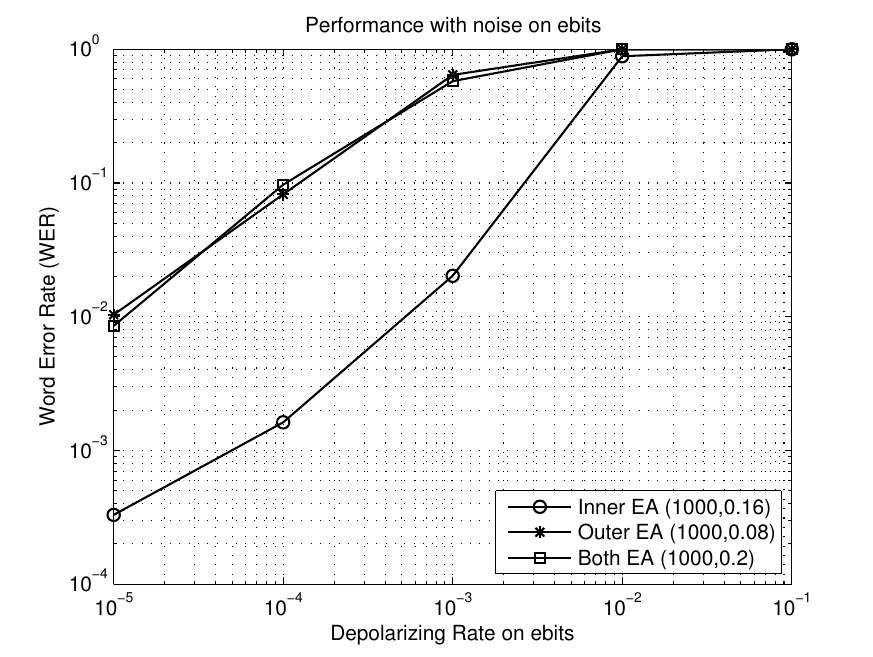}%
\caption{The figure again plots the performance of different inner and outer encoder
combinations as a function of noise rate on Bob's half of the ebits, though we increased
the number of logical qubits in each combination to be 1000. Interestingly,
the combination PTO3REA inner / PTO3R outer shows an increased performance
after increasing the number of logical qubits to 1000, at least for ebit noise levels at
$10^{-3}$ or below. The other combinations perform worse after increasing the number of
logical qubits to 1000.}%
\label{fig:noise-on-ebits-1000}%
\end{center}
\end{figure}

\section{Recursive, Classically-Enhanced Subsystem Encoders are Catastrophic}

\label{sec:corollary}We can construct other variations of quantum convolutonal
encoders and study their state diagrams to determine their properties. One
example of a variation mentioned at the end of Ref.~\cite{PTO09}\ is a
subsystem convolutional code (based on the idea of a subsystem quantum
code~\cite{kribs:180501,qic2006kribs,poulin:230504} that is useful in
fault-tolerant quantum computation~\cite{aliferis:220502}). Poulin \textit{et
al}.~suggest that subsystem convolutional codes might be \textquotedblleft a
concrete avenue\textquotedblright\ for circumventing the inability of a
quantum convolutional encoder to be simultaneously recursive and
non-catastrophic. Such codes exploit a resource known as a \textquotedblleft
gauge\textquotedblright\ qubit that can add extra degeneracy beyond that
available in a standard stabilizer code. Another variation is an encoder that
encodes both classical bits and qubits, and we might wonder if these could be
simultaneously recursive and non-catastrophic (such codes are known as
\textquotedblleft classically-enhanced\textquotedblright%
\ codes~\cite{kremsky:012341} and are based on trade-off coding ideas from
quantum Shannon theory~\cite{DS05}).

Unfortunately, encoders that act on logical qubits, classical bits, ancilla
qubits, and gauge qubits cannot possess both properties simultaneously, and we
state this result below as a corollary of Theorem~1 in Ref.~\cite{PTO09}. This
result implies that entanglement is \textit{the} resource enabling a
convolutional encoder to be both recursive and non-catastrophic (there are no
other known local resources for quantum codes besides ancilla qubits,
classical bits, and gauge qubits).

\begin{corollary}
\label{cor:ce-sub-r-nc}Suppose that a classically-enhanced subsystem
convolutional encoder is recursive. Then it is catastrophic.
\end{corollary}

\begin{IEEEproof}
First, consider that the state diagram for a classically-enhanced subsystem
quantum convolutional encoder includes an edge from $M$ to $M^{\prime}$ if
there exists a $k_{q}$-qubit Pauli operator $L_{q}\in\left\{  I,X,Y,Z\right\}
^{k_{q}}$, an $a$-qubit Pauli operator $S^{z}\in\left\{  I,Z\right\}  ^{a}$, a
$k_{c}$-qubit Pauli operator $L_{c}^{x}+C^{z}\in\left\{  I,X,Y,Z\right\}
^{k_{c}}$, and a $g$-qubit Pauli operator $G^{x}+G^{z}\in\left\{
I,X,Y,Z\right\}  ^{g}$ such that%
\begin{equation}
\left(  M^{\prime}:P\right)  =(M:L_{q}:S^{z}:L_{c}^{x}+C^{z}:G^{x}+G^{z})U,
\label{eq:unified-code-diagram}%
\end{equation}
where the binary representations of Pauli operators are the same as in
(\ref{eq:state-diagram}). We break the operator acting on the $k_{c}$
classical bits into two parts because $L_{c}^{x}$ represents the $X$ component
of the operator that can flip a classical bit from $\left\vert 0\right\rangle
$ to $\left\vert 1\right\rangle $ and back, while $C^{z}$ represents the $Z$
component of the operator that has no effect on a classical bit in state
$\left\vert 0\right\rangle $ or $\left\vert 1\right\rangle $ (it merely adds
an irrelevant global phase in the case that the bit is equal to $\left\vert
1\right\rangle $). The entries $G^{x}$ and $G^{z}$ are for the
\textquotedblleft gauge qubits\textquotedblright\ in a subsystem
code~\cite{poulin:230504} (qubits in a maximally mixed state that are
invariant under the random application of an arbitrary Pauli operator). We
include these transitions in the state diagram because a particular logical
operator $L\equiv\left(  L_{q}:L_{c}^{x}\right)  $ in a classically-enhanced
subsystem code is equivalent up to $Z$ operators acting on the ancilla qubits,
$Z$ operators acting on the classical bits, and $X$ and $Z$ operators acting
on gauge qubits. The set $C\left(  L\right)  $ corresponding to a particular
logical input $L\equiv\left(  L_{q}:L_{c}^{x}\right)  $ consists of all
operators $P$\ of the following form:%
\[
C\left(  L\right)  \equiv\left\{
\begin{array}
[c]{c}%
P:P=(L_{q}:S^{z}:L_{c}^{x}+C^{z}:G^{x}+G^{z})V,\\
S^{z}\in\left\{  I,Z\right\}  ^{a},\ \ C^{z}\in\left\{  I,Z\right\}  ^{k_{c}%
},\\
G^{x}+G^{z}\in\left\{  I,X,Y,Z\right\}  ^{g}%
\end{array}
\right\}  ,
\]
where $V$ in this case consists of the infinite, repeated, overlapping
application of the convolutional encoder $U$.

Suppose that the encoder is recursive. By definition, it follows that every
weight-one logical input $L$ and its coset $C\left(  L\right)  $ has an
infinite response. Suppose now that we change the resources in the encoder so
that all of the classical bits and gauge qubits become ancilla qubits. The
resulting code is now a standard stabilizer code. Additionally, we can remove
the $X$\ operators $L_{c}^{x}$ from the definition of $C\left(  L\right)  $
because they no longer play a role as a logical operator, and we can remove
the $X$\ operators $G^{x}$ from the definition of $C\left(  L\right)  $
because they are acting on ancilla qubits. Let $C^{\prime}\left(  L^{\prime
}\right)  $ denote the new set corresponding to a particular logical operator
$L^{\prime}\equiv\left(  L_{q}:I_{k_{c}}\right)  $:%
\[
C^{\prime}\left(  L^{\prime}\right)  \equiv\left\{
\begin{array}
[c]{c}%
P:P=(L_{q}:S^{z}:C^{z}:G^{z})V,\\
S^{z}\in\left\{  I,Z\right\}  ^{a},\ \ C^{z}\in\left\{  I,Z\right\}  ^{k_{c}%
},\\
G^{z}\in\left\{  I,Z\right\}  ^{g}%
\end{array}
\right\}  .
\]
The encoder is still recursive because $C^{\prime}\left(  L^{\prime}\right)
\subset C\left(  L^{\prime}\right)  $ and because the original encoder is
recursive (recursiveness is a property invariant under the replacement of
cbits and gauge qubits with ancilla qubits). Then Theorem~1 of
Ref.~\cite{PTO09} implies that the encoder for the stabilizer code is
catastrophic, i.e., it features a zero physical-weight cycle with non-zero
logical weight. It immediately follows that the original encoder for the
classically-enhanced subsystem code is catastrophic because its state diagram
contains all of the edges of the stabilizer encoder, plus additional edges
that correspond to the logical transitions for the classical bits and the $X$
operators acting on the gauge qubits.
\end{IEEEproof}

We should note that the above argument only holds if the original
classically-enhanced subsystem convolutional encoder acts on a non-zero number
of logical qubits (this of course is the case in which we are really
interested in order to have a non-zero quantum communication rate). For
example, the encoder in Figure~\ref{fig:simple-encoder}\ becomes recursive and
non-catastrophic when replacing the logical qubit with a classical bit and the
ebit with an ancilla. One can construct the state diagram from the
specification in (\ref{eq:unified-code-diagram}) and discover that these two
properties hold for the encoder acting on the cbit and ancilla.

The above argument also does not apply if an ebit is available as an auxiliary
resource for an encoder. We have found an example of a non-catastrophic,
recursive encoder that acts on six memory qubits, one logical qubit, one
ancilla qubit, one ebit, one cbit, and one gauge qubit. The seed
transformation for this example is as follows:%
\[%
\begin{array}
[c]{c}%
\{2934387,3034525,2323986,3804870,964819,2111270,\\
\ \ \ 3722470,1017852,2678995,912985,1575948,3138834,\\
\ \ \ 142503,317898,3473283,1492935,2701400,2492844,\\
\ \ \ 2078030,1894781,915544,75029\}.
\end{array}
\]

A non-catastrophic, non-recursive subsystem convolutional encoder with a high
free distance could potentially serve well as an outer encoder of a quantum
turbo code, but this might not necessarily be beneficial if our aim is to
achieve the capacity of a quantum channel. The encoder for a subsystem code is
effectively a noisy encoding map because it is a unitary acting on information
qubits in a pure state, ancilla qubits in a pure state, and gauge qubits in
the maximally mixed state. Ref.~\cite{HW08a} proves that an isometric encoder
is sufficient to attain capacity, and we can thus restrict our attention to
subspace codes rather than subsystem codes. Though, this line of reasoning
does not rule out the possibility that iterative decoding could somehow
benefit from the extra degeneracy available in a subsystem code, but this
extra degeneracy would increase the decoding time.

\section{Classically-Enhanced EAQ Encoders}

\label{sec:CE-EAQ}

We might also wish to construct classically-enhanced EAQ\ turbo codes that
transmit classical information in addition to quantum information, in an
effort to reach the optimal trade-off rates from quantum Shannon
theory~\cite{HW08a,HW09}. These codes are then based on the structure of the
codes in Refs.~\cite{kremsky:012341,WB08IEEE}. The state diagram for the
encoder includes an edge from $M$ to $M^{\prime}$ if there exists a $k_{q}%
$-qubit Pauli operator $L_{q}\in\left\{  I,X,Y,Z\right\}  ^{k_{q}}$, a $k_{c}%
$-qubit Pauli operator $L_{c}^{x}\in\left\{  I,X\right\}  ^{k_{q}}$, a $k_{c}%
$-qubit Pauli operator $L_{c}^{z}\in\left\{  I,Z\right\}  ^{k_{q}}$, and an
$a$-qubit Pauli operator $S^{z}\in\left\{  I,Z\right\}  ^{a}$ such that%
\[
\left(  M^{\prime}:P\right)  =(M:L_{q}:S^{z}:L_{c}^{x}+L_{c}^{z}:I_{c})U.
\]
The logical operator $L_{q}$ acts on $k_{q}$\ qubits while the operator
$L_{c}^{x}+L_{c}^{z}$ acts on $k_{c}$ classical bits. We break this latter
operator into two parts for the same reason discussed in the previous section.
The state diagram for an encoder of this form is similar to that for an
EAQ\ convolutional encoder because it includes all transitions for the
classical bits. The difference is in the interpretation of the logical weight
of edge transitions and in the logical label of an edge. If an edge features a
$Z$ operator acting on a classical bit, then this operator does not contribute
to the logical weight of the transition and does not appear on the logical
labels---an $I$ appears on the logical label if $I$ or $Z$ acts on the
classical bit and an $X$ appears on the logical label if an $X$ or $Y$ acts on it.

We have found several examples of recursive, non-catastrophic encoders acting
on these resources. One of our examples acts on five memory qubits, one
logical qubit, one ancilla qubit, one ebit, and one classical bit. Its seed
transformation is as follows:%
\[%
\begin{array}
[c]{c}%
\{162568,174333,260907,261622,168521,3644,\\
101236,204504,227699,260074,138781,9252,\\
257257,179548,78687,6352,230925,180362\}.
\end{array}
\]

\section{The Preservation of Recursiveness and Non-Catastrophicity under
Resource Substitution}

\label{sec:preserve}

The technique used to prove Corollary~\ref{cor:ce-sub-r-nc}\ motivates us to
consider which resource substitutions preserve the properties of recursiveness
and non-catastrophicity. We have two different cases:

\begin{enumerate}
\item A resource substitution that removes edges from the state diagram
preserves non-catastrophicity and recursiveness.

\item A resource substitution that adds edges to the state diagram preserves
catastrophicity and non-recursiveness.
\end{enumerate}

To see the first case for recursiveness preservation, consider a general
encoder acting on logical qubits, ancilla qubits, ebits, cbits, and gauge
qubits. Each logical input $L\equiv\left(  L_{q}:L_{c}^{x}\right)  $ to the
encoder has the following form:%
\[
C\left(  L\right)  \equiv\left\{
\begin{array}
[c]{l}%
P:\\
P=(L_{q}:S^{z}:I_{c}:L_{c}^{x}+C^{z}:G^{x}+G^{z})V,\\
S^{z}\in\left\{  I,Z\right\}  ^{a},\ \ C^{z}\in\left\{  I,Z\right\}  ^{k_{c}%
},\\
G^{x}+G^{z}\in\left\{  I,X,Y,Z\right\}  ^{g}%
\end{array}
\right\}  ,
\]
where the conventions are similar to what we had before. Suppose the encoder
is recursive so that $C\left(  X_{i}:I\right)  $, $C\left(  Y_{i}:I\right)  $,
$C\left(  Z_{i}:I\right)  $, and $C\left(  I:X_{j}\right)  $ all have an
infinite response. Then the encoder is still recursive if we replace an
ancilla with an ebit because the original encoder is recursive and we no
longer have to consider the $Z$ operator acting on the replaced ancilla. The
encoder is still recursive if we replace a cbit with an ancilla qubit because
we no longer have to consider the coset $C\left(  I:X_{j}\right)  $. Finally,
it is still recursive if we replace a gauge qubit with an ancilla qubit
because we no longer have to consider the $X$ operator acting on the replaced
gauge qubit.

To see the first case for non-catastrophicity preservation, suppose that an
encoder is already non-catastrophic. Then a resource substitution that removes
edges from the state diagram preserves non-catastrophicity because this
removal cannot create a zero physical-weight cycle with non-zero logical weight.

To see the second case for non-recursiveness preservation, suppose that an
encoder acting on logical qubits and ebits is non-recursive, meaning that at
least one of the weight-one logical inputs $X_{i}$, $Y_{i}$, or $Z_{i}$ has a
finite response. Then replacing an ebit with an ancilla certainly cannot make
the resulting encoder recursive because we have to consider the response to
the operators $X_{i}$, $Y_{i}$, or $Z_{i}$ and one of these is already finite
from the assumption of non-recursiveness of the original encoder. Furthermore,
consider a non-recursive encoder acting on logical qubits, ebits, and ancilla
qubits. Replacing some of the ancilla qubits with cbits or gauge qubits cannot
make the encoder become recursive for the same reasons.

To see the second for catastrophicity preservation, suppose that an encoder is
catastrophic. Then a resource substitution that adds edges to the state
diagram preserves catastrophicity because any zero physical-weight cycles with
non-zero logical weight are still part of the state diagram for the new encoder.

The following diagram summarizes all of the above observations. The resources
are $L$ for logical qubits, $S$ for ancilla qubits, $E$ for ebits, $C$ for
cbits, and $G$ for gauge qubits. Recursiveness and non-catastrophicity
preservation flow downwards under the displayed resource substitutions, while
non-recursiveness and catastrophicity preservation flow upwards (substitutions
at the same level can go in any order).

\begin{center}%
\begin{tabular}
[c]{ccccccccccc}%
$($ & $L$ & $:$ & $S$ & $:$ & $E$ & $:$ & $C$ & $:$ & $G$ & $)$\\
&  &  & $\downarrow\uparrow$ &  &  &  & $\downarrow\uparrow$ &  &
$\downarrow\uparrow$ & \\
$($ & $L$ & $:$ & $E$ & $:$ & $E$ & $:$ & $S$ & $:$ & $S$ & $)$\\
&  &  &  &  &  &  & $\downarrow\uparrow$ &  & $\downarrow\uparrow$ & \\
$($ & $L$ & $:$ & $E$ & $:$ & $E$ & $:$ & $E$ & $:$ & $E$ & $)$%
\end{tabular}

\end{center}

We can then understand the proof of Corollary~\ref{cor:ce-sub-r-nc} in the
context of the above diagram. The original encoder acts on $\left(
L:S:C:G\right)  $ and is assumed to be recursive. Resource substitution of $C$
and $G$ preserves recursiveness, while making the encoder act on only ancilla
qubits (the auxiliary resource for a standard quantum code). Using the fact
that standard recursive quantum encoders are catastrophic~\cite{PTO09}, we
then back substitute the resources, which is a catastrophicity preserving substitution.

\section{Conclusion and Current Work}

We have constructed a theory of EAQ\ serial turbo coding as an extension of
Poulin \textit{et al}.'s theory in Ref.~\cite{PTO09}. The introduction of
shared entanglement simplifies the theory because an EAQ\ convolutional
encoder can be both recursive and non-catastrophic. These two properties are
essential for quantum serial turbo code families to have minimum distance that
grows near-linearly with the length of the code, while still performing well
under iterative decoding. We provided many examples of EAQ\ convolutional
encoders that satisfy both properties, and we detailed their parameters. We
then showed how the concatenation of these encoders with some of Poulin
\textit{et al}.'s and some of our own lead to EAQ\ serial turbo codes with
near-linear minimum distance scaling. We modified the quantum turbo decoding
algorithm from Ref.~\cite{PTO09} such that it follows the turbo decoding principle
in which the constituent decoders pass along extrinsic information, and this
modification lead to a significant performance improvement over the algorithm
outlined in Ref.~\cite{PTO09}.
We conducted several simulations of EAQ\ turbo
codes---several of our quantum turbo codes were within 1~dB
of their hashing limits and two notable surprises
were that placing ebits in the inner encoder can
achieve a better performance than expected in both scenarios with and without
ebit noise. Our simulations are generally consistent with the findings in
Ref.~\cite{WB09,LB10}\ and other results from quantum Shannon
theory~\cite{PhysRevLett.83.3081,arx2005dev}, namely, that entanglement
assistance can significantly enhance error correction ability. Finally, we
considered how to construct the state diagram for encoders that derive from
other existing extensions to the theory of quantum error correction, and we
showed that classically-enhanced subsystem convolutional encoders cannot
simultaneously be recursive and non-catastrophic.

There are many questions to ask going forward from here. One could certainly
seek out other entanglement-assisted quantum turbo codes
and conduct numerical simulations of
their performance. One purpose of our numerical
simulations was to illustrate
the effect of adding entanglement assistance
to the encoders of Poulin {\it et al.} \cite{PTO09},
and it was not our intent for them to constitute an exhaustive code comparison.
It is ongoing work to search for and test many other code combinations,
including cases where the inner and outer encoder
are either recursive, non-recursive, have high / medium / low minimum distance, and have
varying numbers of memory qubits $\leq 4$.
Also, if one wished to compare directly the performance of entanglement-assisted turbo codes against
the codes in Ref.~\cite{PTO09},
one way to do so might be to look for higher-rate codes operating near the entanglement-assisted
hashing bound that tolerate the same noise level as the codes from Ref.~\cite{PTO09}.
One could also vary the number of ebits and ancilla qubits present in the various code combinations
discussed here.

It would be
interesting to explore the performance of the other suggested code structures
in Section~\ref{sec:CE-EAQ} to determine if they could come close to achieving
the optimal rates from quantum Shannon theory~\cite{HW08a,HW09}. For example,
what is the best arrangement for a classically-enhanced EAQ code? Should we
place the classical bits in the inner or outer encoder? Are there more clever
ways to use entanglement so that we increase error-correcting ability while
reducing entanglement consumption? Consider that Hsieh \textit{et
al}.~recently constructed a class of entanglement-assisted codes that exploit
one ebit and still have good performance~\cite{HYH09}.

We should stress that the behavior of the entanglement-assisted codes using
maximal entanglement is exactly like that of a classical turbo code. There is
no degeneracy, and the iterative decoding algorithm is exactly the same as the
classical one. Furthermore, analyses of classical turbo codes should apply
directly in these cases~\cite{BDMP98,JM02}, and it would be good to determine
the exact correspondence. These analyses studied the bit error rate rather
than the word error rate, so any study of EAQ turbo codes would have to factor
into account this difference.

Much of the classical literature has focused on the choice of a practical
interleaver rather than a random one~\cite{BP94,DD95,YVF99,SSSN01}, and it
might be interesting to import the knowledge discovered here to the quantum case.

Finally, it would be great to find examples of EAQ\ turbo codes with a
positive catalytic rate that outperform the turbo codes in Ref.~\cite{PTO09},
in the sense that they either have a higher catalytic rate while tolerating
the same noise levels or they have the same catalytic rate while tolerating
higher noise levels. This is ongoing work.

We acknowledge David Poulin for providing us with a copy of Ref.~\cite{OPT08}%
\ and for many useful discussions, e-mail interactions, and feedback on the
manuscript. We acknowledge Jean-Pierre Tillich for originally suggesting that
shared entanglement might help in quantum serial turbo codes. We acknowledge
Todd Brun, Hilary Carteret, and Jan Florjanczyk for useful discussions, and
Patrick Hayden for the observation that non-degenerate quantum codes lead to
entanglement-assisted codes with particular error correction power on Bob's
half of the ebits. Zunaira Babar is grateful to
Prof.~Lajos Hanzo and Dr.~Soon Xin Ng (Michael) for their
continuous guidance and support. We acknowledge the computer administrators in the McGill
School of Computer Science for making their computational resources available
for this scientific research. MMW\ acknowledges the warm hospitality of the
ERATO-SORST\ project, the support of the MDEIE (Qu\'{e}bec) PSR-SIIRI
international collaboration grant, and the support of the Centre de Recherches
Math\'ematiques in Montreal.

\appendices

\section{Computing the Distance Spectrum}

\label{sec:distance-spectrum}

There is a straightforward way to compute the distance spectrum of a quantum convolutional
encoder. This technique borrows from similar ideas in the classical theory of
convolutional coding \cite{V71,M98,book1999conv,McE02}. We would like to know
the number of admissible paths with a particular weight beginning and ending
in memory states that are part of a zero physical-weight cycle. For our
example in Section~\ref{sec:examples}, the identity memory state is the only memory state part of a zero
physical-weight cycle. We create a \textit{weight adjacency matrix} whose
entries correspond to edges in the state diagram. This matrix has $x^{w}$ in
entry $\left(  i,j\right)  $ if there is a physical-weight-$w$ edge from
vertex $i$ to vertex $j$ (with the exception of the self-loop at the identity
memory state). The weight adjacency matrix $A$\ for our example is%
\[
A\equiv%
\begin{bmatrix}
0 & x^{2} & x & x\\
x^{2} & x^{2} & x^{2} & x^{2}\\
x^{2} & x & x & x^{2}\\
x^{2} & x & x^{2} & x
\end{bmatrix}
,
\]
where the ordering of vertices is $I$, $X$, $Y$, and $Z$. Note that we place a
zero in the $\left(  1,1\right)  $ entry because we do not want to overcount
the number of admissable paths starting and ending in memory states that are
part of a zero physical-weight cycle. If we would like the number of
admissable paths up to an arbitrary weight that start and end in the identity
memory state, then we compute the $\left(  1,1\right)  $ entry of the
following matrix:%
\[
\left(  I-A\right)  ^{-1}=I+A+A^{2}+A^{3}+\cdots.
\]
The coefficient of $x^{w}$ in the polynomial entry
$\left(  I-A\right)  ^{-1}\left(  1,1\right)  $ is
the number of admissable paths with weight $w$\ starting and ending in the
identity memory state. One can compute this in some cases using Cramer's rule, for example.
We can also approximate the distance spectrum, e.g., by
computing the matrix $B$ where%
\[
B=\sum_{i=1}^{T}A^{i},
\]
and $T$ is some finite positive integer, so that this approximation gives a
truncated distance spectrum. 
Computing the above matrix can be computationally expensive for large $T$, but we can
dramatically reduce the number of computations by truncating the polynomial
entries of $A$ above degree $T$ before performing each multiplication. For our
example in Section~\ref{sec:examples}, the first ten entries of the distance spectrum polynomial are%
\[
2x^{3}+5x^{4}+6x^{5}+23x^{6}+54x^{7}+122x^{8}+298x^{9}+737x^{10},
\]
so that this gives a fairly reasonable approximation to the true distance spectrum.
These coefficients appear in the second column of
Table~\ref{tbl:distance-spec}\ as the first ten values of the distance
spectrum for this first example encoder. Note that there are faster ways of
computing the distance spectrum for classical convolutional codes~\cite{CJ89},
and it remains open to determine how to exploit these techniques for quantum
convolutional encoders.%

\section{Example Encoders}

\label{sec:examples-of-encoders}

\begin{table*}[tbp] \centering
\begin{tabular}
[c]{l|l|l|l|l|l|l}\hline\hline
\textbf{Encoder} & \textbf{M} & \textbf{L} & \textbf{A} & \textbf{E} &
\textbf{Seed Transformation} & \textbf{Free Dist.}\\\hline\hline
\multicolumn{1}{c|}{1} & 1 & 1 & 0 & 1 & $\{33,29,30,7,45,47\}$ &
\multicolumn{1}{|c}{3}\\\hline
\multicolumn{1}{c|}{2} & 3 & 2 & 0 & 1 & $\left\{
2188,246,115,2053,1847,833,1658,2571,1566,2783,2990,3229\right\}  $ &
\multicolumn{1}{|c}{4}\\\hline
\multicolumn{1}{c|}{3} & 3 & 3 & 0 & 1 &
$\{12515,8790,10280,11314,6500,14691,1430,7105,8817,1420,10014,7061,10739,8972\}$
& \multicolumn{1}{|c}{4}\\\hline
\multicolumn{1}{c|}{4} & 3 & 4 & 0 & 1 & $%
\begin{array}
[c]{c}%
\{23233,28350,13963,43904,58908,19553,6318,63573,\\
12838,7558,22611,27045,48320,9596,48500,54018\}
\end{array}
$ & \multicolumn{1}{|c}{3}\\\hline
\multicolumn{1}{c|}{5} & 2 & 1 & 1 & 1 &
$\{159,1006,727,641,925,522,726,314,793,648,119,210\}$ &
\multicolumn{1}{|c}{4}\\\hline
\multicolumn{1}{c|}{6} & 2 & 1 & 1 & 2 &
$\{1116,1363,1495,1326,241,2411,2268,1480,2032,1589,810,3351\}$ &
\multicolumn{1}{|c}{5}\\\hline
\multicolumn{1}{c|}{7} & 2 & 2 & 1 & 1 &
$\{141,509,3495,2470,2702,3576,1522,905,2622,1598,642,773\}$ &
\multicolumn{1}{|c}{3}\\\hline
\multicolumn{1}{c|}{8} & 2 & 6 & 0 & 1 & $%
\begin{array}
[c]{c}%
\{113633,199924,181760,243189,25748,110950,158559,282,205474,\\
193680,199692,252779,245067,64266,147306,152171,230343,75396\}
\end{array}
$ & \multicolumn{1}{|c}{2}\\\hline
\multicolumn{1}{c|}{9} & 2 & 8 & 0 & 1 & $%
\begin{array}
[c]{c}%
\{2432999,1503627,1816960,1050871,1297694,3894582,410463,2344289,1908709,\\
3176421,3668357,1860207,1511167,3829280,3008050,2896381,999389,374648,\\
4000734,885953,2452389,3608225\}
\end{array}
$ & \multicolumn{1}{|c}{2}\\\hline
\multicolumn{1}{c|}{10} & 2 & 9 & 0 & 1 & $%
\begin{array}
[c]{c}%
\{4943947,12156608,10237254,2501342,2665695,7306816,8727132,80870,13726997,\\
16078090,11897398,9857749,16524053,972786,5098459,8962232,10325041,\\
12705543,8324846,13241728,11521711,7907747,16588769,5842661\}
\end{array}
$ & \multicolumn{1}{|c}{N/A}\\\hline\hline
\end{tabular}
\caption{Specification of our example EAQ convolutional encoders that are recursive and non-catastrophic.
The first column indexes the different encoders, and all of the other columns
correspond to a particular encoder.
Column ``M'' gives the number of memory qubits, column ``L'' gives
the number of information qubits, column ``A'' gives the number
of ancilla qubits, and column ``E'' gives the number of ebits.
The column labeled ``Seed Transformation'' specifies the seed transformation for each encoder,
using a decimal representation. The convention for qubit ordering is that an encoder acts on memory qubits,
information qubits, ancilla qubits, and ebits to produce output memory qubits
and physical or channel qubits. The convention for the bit representation of a tensor product of Pauli
operators is $ZXYZ$ maps to $[1011|0110]$, and the decimal representation of this bit vector is 182
(see the discussion around (\ref{eq:example-encoder}) in the text).
The last column gives the free distance of the encoder.}\label{tbl:example-encoders}%
\end{table*}%
Table~\ref{tbl:example-encoders} lists the specifications of many other
examples of EAQ encoders that are both recursive and non-catastrophic---a
computer program helped check that these properties hold for each of the
examples~\cite{W10}. Included in the list of example encoders are some which
act on ancilla qubits in addition to ebits. These examples demonstrate that we
do not necessarily require the auxiliary resource of an EAQ\ convolutional
encoder to be ebits alone in order for the encoder to possess both properties.
Table~\ref{tbl:distance-spec} gives a truncated distance spectrum for each of these encoders.%

\begin{table}[tbp] \centering
\begin{tabular}
[c]{l|c|c|c|c|c|c|c}\hline\hline
$w$ & \textbf{1} & \textbf{2} & \textbf{3} & \textbf{4} & \textbf{5} &
\textbf{6} & \textbf{7}\\\hline\hline
\multicolumn{1}{c|}{0} & \multicolumn{1}{|r|}{0} & \multicolumn{1}{|r|}{0} &
\multicolumn{1}{|r|}{0} & \multicolumn{1}{|r|}{0} & \multicolumn{1}{|r|}{0} &
\multicolumn{1}{|r|}{0} & \multicolumn{1}{|r}{0}\\
\multicolumn{1}{c|}{1} & \multicolumn{1}{|r|}{0} & \multicolumn{1}{|r|}{0} &
\multicolumn{1}{|r|}{0} & \multicolumn{1}{|r|}{0} & \multicolumn{1}{|r|}{0} &
\multicolumn{1}{|r|}{0} & \multicolumn{1}{|r}{0}\\
\multicolumn{1}{c|}{2} & \multicolumn{1}{|r|}{0} & \multicolumn{1}{|r|}{0} &
\multicolumn{1}{|r|}{0} & \multicolumn{1}{|r|}{0} & \multicolumn{1}{|r|}{0} &
\multicolumn{1}{|r|}{0} & \multicolumn{1}{|r}{0}\\
\multicolumn{1}{c|}{3} & \multicolumn{1}{|r|}{2} & \multicolumn{1}{|r|}{0} &
\multicolumn{1}{|r|}{0} & \multicolumn{1}{|r|}{3} & \multicolumn{1}{|r|}{0} &
\multicolumn{1}{|r|}{0} & \multicolumn{1}{|r}{3}\\
\multicolumn{1}{c|}{4} & \multicolumn{1}{|r|}{5} & \multicolumn{1}{|r|}{1} &
\multicolumn{1}{|r|}{8} & \multicolumn{1}{|r|}{32} & \multicolumn{1}{|r|}{3} &
\multicolumn{1}{|r|}{0} & \multicolumn{1}{|r}{22}\\
\multicolumn{1}{c|}{5} & \multicolumn{1}{|r|}{6} & \multicolumn{1}{|r|}{6} &
\multicolumn{1}{|r|}{69} & \multicolumn{1}{|r|}{292} & \multicolumn{1}{|r|}{3}
& \multicolumn{1}{|r|}{1} & \multicolumn{1}{|r}{73}\\
\multicolumn{1}{c|}{6} & \multicolumn{1}{|r|}{23} & \multicolumn{1}{|r|}{49} &
\multicolumn{1}{|r|}{463} & \multicolumn{1}{|r|}{2,622} &
\multicolumn{1}{|r|}{23} & \multicolumn{1}{|r|}{1} & \multicolumn{1}{|r}{286}%
\\
\multicolumn{1}{c|}{7} & \multicolumn{1}{|r|}{54} & \multicolumn{1}{|r|}{218}
& \multicolumn{1}{|r|}{3,478} & \multicolumn{1}{|r|}{24,848} &
\multicolumn{1}{|r|}{41} & \multicolumn{1}{|r|}{1} &
\multicolumn{1}{|r}{1,309}\\
\multicolumn{1}{c|}{8} & \multicolumn{1}{|r|}{122} &
\multicolumn{1}{|r|}{1,077} & \multicolumn{1}{|r|}{25,057} &
\multicolumn{1}{|r|}{227,262} & \multicolumn{1}{|r|}{127} &
\multicolumn{1}{|r|}{3} & \multicolumn{1}{|r}{5,696}\\
\multicolumn{1}{c|}{9} & \multicolumn{1}{|r|}{298} &
\multicolumn{1}{|r|}{5,477} & \multicolumn{1}{|r|}{181,959} &
\multicolumn{1}{|r|}{2.1$\cdot10^{6}$} & \multicolumn{1}{|r|}{325} &
\multicolumn{1}{|r|}{11} & \multicolumn{1}{|r}{23,975}\\
\multicolumn{1}{c|}{10} & \multicolumn{1}{|r|}{737} &
\multicolumn{1}{|r|}{27,428} & \multicolumn{1}{|r|}{1,326,070} &
\multicolumn{1}{|r|}{1.9$\cdot10^{7}$} & \multicolumn{1}{|r|}{1,061} &
\multicolumn{1}{|r|}{17} & \multicolumn{1}{|r}{102,132}\\\hline\hline
\end{tabular}
\caption{A truncated distance spectrum for the first seven of our example EAQ
convolutional encoders. The first column on the left
gives the first 11 values of the weight $w$, and the other columns give the first 11 values of
the truncated distance
spectrum $F(w)$ of the first seven of our EAQ convolutional encoders from Table~\ref{tbl:example-encoders}.}\label{tbl:distance-spec}%
\end{table}%

We can construct EAQ\ serial turbo codes, by serially concatenating some of
our example \textquotedblleft WH\ encoders\textquotedblright\ in
Table~\ref{tbl:example-encoders} with the \textquotedblleft
PTO\ encoders\textquotedblright\ in Table~1 of Ref.~\cite{PTO09} (ordered from
left to right).\ Table~\ref{tbl:turbo-combos} details these different
combinations, giving their rates and average minimum distance scaling.%
\begin{table}[tbp] \centering
\begin{tabular}
[c]{c|c|c|c|c}\hline\hline
\textbf{Outer} & \textbf{Inner} & \textbf{Q} & \textbf{E} & \textbf{Min.
Dist.}\\\hline\hline
PTO1 & WH3 & 1/4 & 1/4 & $N^{2/3}$\\
PTO2 & WH3 & 1/4 & 1/4 & $N^{5/7}$\\
PTO3 & WH2 & 1/3 & 1/3 & $N^{3/4}$\\
PTO3 & WH7 & 1/4 & 1/4 & $N^{3/4}$\\
PTO3 & WH4 & 2/5 & 1/5 & $N^{3/4}$\\
PTO2 & WH8 & 2/7 & 1/7 & $N^{5/7}$\\
PTO3 & WH8 & 3/7 & 1/7 & $N^{3/4}$\\
PTO3 & WH9 & 4/9 & 1/9 & $N^{3/4}$\\
PTO2 & WH10 & 3/10 & 1/10 & $N^{5/7}$\\
WH11 & WH3 & 1/2 & 1/4 & $N^{1/3}$\\\hline\hline
\end{tabular}
\caption{Various combinations of encoders for quantum turbo codes and their minimum distance scaling.}\label{tbl:turbo-combos}%
\end{table}%

The columns of Table~\ref{tbl:turbo-combos} give the outer encoder, the inner
encoder, the quantum communication rate, the entanglement consumption rate,
and the average minimum distance growth of a particular EAQ\ turbo code. The
first four combinations all have a good average minimum distance scaling, but
the catalytic rate\footnote{The catalytic rate is the difference between the
quantum communication rate and the entanglement consumption rate~\cite{DBH09}%
.} for each of them is zero. The last six combinations both have a good
average minimum distance scaling and a positive catalytic rate.

\bibliographystyle{IEEEtran}
\bibliography{Ref}

\end{document}